\ifdefined\ACCESS
  \documentclass{ieeeaccess}
\else
  \documentclass[journal]{IEEEtran}
\fi
\usepackage{cite}
\usepackage{amsmath,amssymb,amsfonts,amsthm,mathtools,bm}
\usepackage{algorithm}
\usepackage{algpseudocode}
\usepackage{graphicx}
\usepackage{textcomp}
\usepackage{booktabs,multirow,array}
\usepackage{siunitx}
\usepackage{enumitem}
\usepackage[nameinlink,noabbrev]{cleveref}
\usepackage{microtype}
\usepackage{url}

\graphicspath{{figures/}}
\newdimen\xfigwd
\xfigwd=\columnwidth
\def\BibTeX{{\rm B\kern-.05em{\sc i\kern-.025em b}\kern-.08em
    T\kern-.1667em\lower.7ex\hbox{E}\kern-.125emX}}

\newcommand{\M}{\mathbf{M}}
\newcommand{\K}{\mathbf{K}}
\newcommand{\C}{\mathbf{C}}
\newcommand{\A}{\mathbf{A}}
\newcommand{\Hq}{\mathbf{H}}
\newcommand{\I}{\mathbf{I}}
\newcommand{\PhiM}{\boldsymbol{\Phi}}
\newcommand{\ph}{\boldsymbol{\phi}}
\newcommand{\yy}{\mathbf{y}}
\newcommand{\uu}{\mathbf{u}}
\newcommand{\ff}{\mathbf{f}}
\newcommand{\rr}{\mathbf{r}}
\newcommand{\norm}[1]{\left\lVert #1\right\rVert}
\newcommand{\tr}{\operatorname{tr}}
\newcommand{\Var}{\operatorname{Var}}
\newcommand{\diag}{\operatorname{diag}}
\newcommand{\MAC}{\operatorname{MAC}}
\newcommand{\epoch}{\mathrm{ep}}

\newtheorem{theorem}{Theorem}
\newtheorem{proposition}{Proposition}

\newtheorem{remark}{Remark}

\ifdefined\ACCESS\else
  \newcommand{\PARstart}{\IEEEPARstart}
  \newenvironment{keywords}{\begin{IEEEkeywords}}{\end{IEEEkeywords}}
\fi

\crefname{figure}{Fig.}{Figs.}
\Crefname{figure}{Fig.}{Figs.}
\crefname{table}{Table}{Tables}
\Crefname{table}{Table}{Tables}
\crefname{algorithm}{Algorithm}{Algorithms}
\Crefname{algorithm}{Algorithm}{Algorithms}
\crefname{remark}{Remark}{Remarks}
\Crefname{remark}{Remark}{Remarks}
\crefname{proposition}{Proposition}{Propositions}
\Crefname{proposition}{Proposition}{Propositions}
\crefname{section}{Section}{Sections}
\Crefname{section}{Section}{Sections}

\begin{document}
\ifdefined\ACCESS
\history{}
\doi{}
\vol{7}
\year{2026}
\fi

\title{Qubits Are Not Enough: Deflation and Measurement Limits in Variational Quantum Modal Analysis}

\ifdefined\ACCESS
\author{\uppercase{Bikalpa Gautam}\authorrefmark{1}, \uppercase{Akil Raj Subedi}\authorrefmark{2}}
\address[1]{Department of Civil, Environmental and Geomatic Engineering, ETH Zurich, 8093 Zurich, Switzerland}
\address[2]{Zynga Inc., Bangalore, India}

\markboth
{B. Gautam and A. R. Subedi: Qubits Are Not Enough: Deflation and Measurement Limits}
{B. Gautam and A. R. Subedi: Qubits Are Not Enough: Deflation and Measurement Limits}

\corresp{Corresponding author: Bikalpa Gautam (bgautam@ethz.ch)}
\else
\author{Bikalpa~Gautam and Akil~Raj~Subedi%
\thanks{B. Gautam is with the Department of Civil, Environmental and Geomatic
Engineering, ETH Zurich, 8093 Zurich, Switzerland (e-mail: bgautam@ethz.ch).}%
\thanks{A. R. Subedi is with Zynga Inc., Bangalore, India.}}
\fi

\ifdefined\ACCESS\else\maketitle\fi

\begin{abstract}
Quantum computers can store a structure's vibration mode compactly: a mode with $N$ degrees of freedom needs only about $\lceil\log_2 N\rceil$ qubits. This paper asks whether that compactness makes quantum modal analysis practical, using the variational quantum eigensolver (VQE) on finite-element structural models. We build a framework for testing this and make two theoretical contributions. First, we show that a known condition for computing higher modes is not only sufficient but also necessary. Second, we bound the error caused by imperfect lower modes. Simulations show that having enough qubits is not enough. Shallow circuits fail even when the state fits, and converting the problem into quantum form makes the operator much denser, especially in two and three dimensions. The main bottleneck is measurement. Even with $10^5$ shots per term, the first-mode frequency error stays at 2.45\%, and reliably locating a 10\% stiffness loss in a small six-degree-of-freedom benchmark would need roughly $10^{8}$ circuit executions. We claim no quantum speedup.
\end{abstract}


\ifdefined\ACCESS
\titlepgskip=-15pt
\maketitle
\fi

\section{Introduction}
\label{sec:intro}

\PARstart{M}{odal} analysis underlies seismic response, wind and pedestrian serviceability, reduced-order simulation, operational modal analysis, finite-element model updating, and vibration-based structural health monitoring \cite{Chopra2020,Ewins2000,FarrarWorden2007,FarrarWorden2012}. For a structure with $N$ unconstrained degrees of freedom (DOFs), the low-frequency modes solve
\begin{equation}
\K\ph_i=\lambda_i\M\ph_i,
\qquad \lambda_i=\omega_i^2,
\qquad f_i=\frac{\omega_i}{2\pi}.
\label{eq:gep}
\end{equation}
Classical Lanczos, subspace-iteration, and contour-based eigensolvers are mature and highly effective. A quantum method must therefore be evaluated against physical residuals, modal orthogonality, measurement cost, state-preparation cost, and end-to-end preprocessing, not against matrix dimension alone.

The variational quantum eigensolver (VQE) minimizes an expectation-valued Rayleigh quotient over a parameterized quantum state \cite{Peruzzo2014,McClean2016}, and is among the most studied near-term algorithms \cite{Preskill2018,Cerezo2021,Tilly2022}. Excited-state extensions include variational quantum deflation (VQD) \cite{Higgott2019} and subspace-search VQE \cite{Nakanishi2019}. Generalized eigenvalue formulations have been developed for finite-element problems \cite{Sato2023}, and structural demonstrations have coupled finite-element software with VQE on simulators and small processors \cite{Liu2024}. More recently, Taliadouros et al.\ used SSVQE and VQD inside Monte Carlo simulation to estimate multiple eigenvalue statistics of stochastic structural systems, including a proof-of-concept execution on IBM Sherbrooke \cite{Taliadouros2025}. A complementary line of work targets the fault-tolerant regime: Lee and Kanno compute natural frequencies and normal modes of classical coupled oscillators by qubitization and phase estimation, constructing block-encoding oracles from the sparse structure of the matrix, proposing an initial-state preparation, and estimating physical-qubit counts and runtimes for a fault-tolerant machine \cite{LeeKanno2025}. The present paper addresses the complementary near-term question---whether the variational route, whose appeal is precisely that it needs no fault tolerance, is limited by encoding, expressibility, or measurement---and finds the binding constraint to be measurement rather than qubit count. Accordingly, neither the classical Cholesky reduction nor the use of VQD/SSVQE for multiple structural eigenvalues is claimed as new here. Higgott et al.\ already introduce the deflated cost \cref{eq:deflatedH} and observe that $\beta_j>\varepsilon_r-\varepsilon_j$ suffices to expose the target mode \cite{Higgott2019}; that sufficiency condition is theirs, not ours. The contribution relative to this literature is the converse, showing the condition is also necessary and extending it to the padded spectrum; the deterministic certified penalty rule; the identification of the penalty that minimizes the approximate-projector certificate; the separation of padding certification from optimizer scaling; and end-to-end accounting of whitening, Pauli, finite-shot, and damage-decision costs. Existing studies establish feasibility, but do not by themselves establish scalable structural utility.

Recent structural-quantum work also explores alternatives outside the variational eigenvalue
route. Xu et al.\ map elastodynamics to Hamiltonian simulation through an energy-conservation
construction and report a superconducting-processor demonstration \cite{Xu2026}; Wagatsuma
et al.\ apply quantum-singular-value-transformation Hamiltonian simulation to simple harmonic
motion of elastic structures \cite{Wagatsuma2026}; and Xu and Hu develop a decomposition-free
potential-energy formulation aimed at avoiding matrix-to-unitary decomposition overhead
altogether, discussing condition-number precision, state preparation and barren plateaus
\cite{XuHu2026}. Together with \cite{LeeKanno2025} these establish that representation, state
preparation and output precision are algorithm-dependent rather than intrinsic to the problem.
The present work addresses the complementary question of whether a sampling-based variational
modal solver remains attractive once whitening, deflation, trainability, Pauli representation
and engineering-level measurement precision are accounted for together; its negative
conclusions are therefore specific to that setting and are not claims about these alternative
routes.

The central premise of this paper is summarized by its title: logarithmic qubit count is not sufficient. A usable structural modal solver must address six distinct questions:
\begin{enumerate}[label=(\roman*)]
\item Is the generalized eigenproblem mapped to a Hermitian operator without changing the physical mass metric?
\item Are padded and previously extracted states excluded with provable, rather than heuristic, penalty conditions?
\item Does whitening preserve or destroy sparsity and measurement economy?
\item Can the ansatz represent the requested eigenspace and can the cost be trained?
\item Can modal observables be estimated at finite-shot precision relevant to engineering decisions?
\item Can repeated monitoring distinguish genuine structural change from numerical mode swapping without biasing the eigensolution?
\end{enumerate}

Two elements of the framework are adopted rather than claimed. The mass-orthonormal Cholesky mapping, with explicit accounting of the classical cost it incurs and the sparsity it destroys, and the two-scale spectral normalization that decouples the optimizer scale from certified safe padding, are both standard constructions used here as the basis for the analysis. Against that basis, four contributions are claimed. Each is proved or measured in this paper; components that are proposed but not evaluated are named separately afterwards rather than counted.
\begin{enumerate}[label=(\roman*)]
\item for higher modes, the converse to the known variational-deflation threshold---showing the condition is necessary as well as sufficient---its extension to the padded spectrum, a deterministic certified penalty rule, and the identification of the penalty $\beta_j^\star$ that minimizes the approximate-projector certificate, which is not the smallest admissible one;
\item an exact closed form for the nonzero Pauli count of a lumped-mass mesh, $(n_q+2)2^{\,n_q-1}$, with a checkable necessary-and-sufficient condition for it to be attained, showing that within the locality-preserving nearest-neighbour grid family studied here it is consistent-mass whitening, and not the spatial dimension or the mesh size, that drives the Pauli representation from the $O(N\log N)$ lumped bound toward the quadratic real-symmetric limit;
\item end-to-end resource accounting \emph{for the whitened variational pipeline}---whitening, Pauli growth, commuting groups, state preparation, finite-shot frequency uncertainty, and damage-decision cost---including the $O(N^{4})$ penalty that fixed \emph{relative} frequency accuracy imposes on the shot budget, and supported by a two- and three-dimensional generalization in which densification worsens, and appears later in $N$, with spatial dimension under the locality-preserving ordering studied here, together with a degeneracy-diagnostic validation unreachable in one dimension;
\item evidence that expressibility and trainability are in tension for this ansatz family: over $2\le n_q\le10$ the gradient variance shows a finite-size decay rate comparable to $2^{-n_q}$ at $L=8$, already at depths far below the parameter-count threshold for universal real-state expressibility.
\end{enumerate}

Two further procedures are developed but deliberately not claimed, because neither is evaluated experimentally: a lexicographic refinement that replaces weighted sums of energy, variance and leakage by a constrained problem, and a post-hoc temporal mode-tracking rule that assigns mode labels after independent eigensolution instead of through continuity penalties inside it. Both are stated in full in the supplementary material rather than in the article, because the evaluation framework is incomplete without them but neither is a validated result.

\Cref{fig:workflow} summarizes the end-to-end pipeline analyzed in the rest of the paper.

\begin{figure*}[t]
\centering
\includegraphics[width=\textwidth]{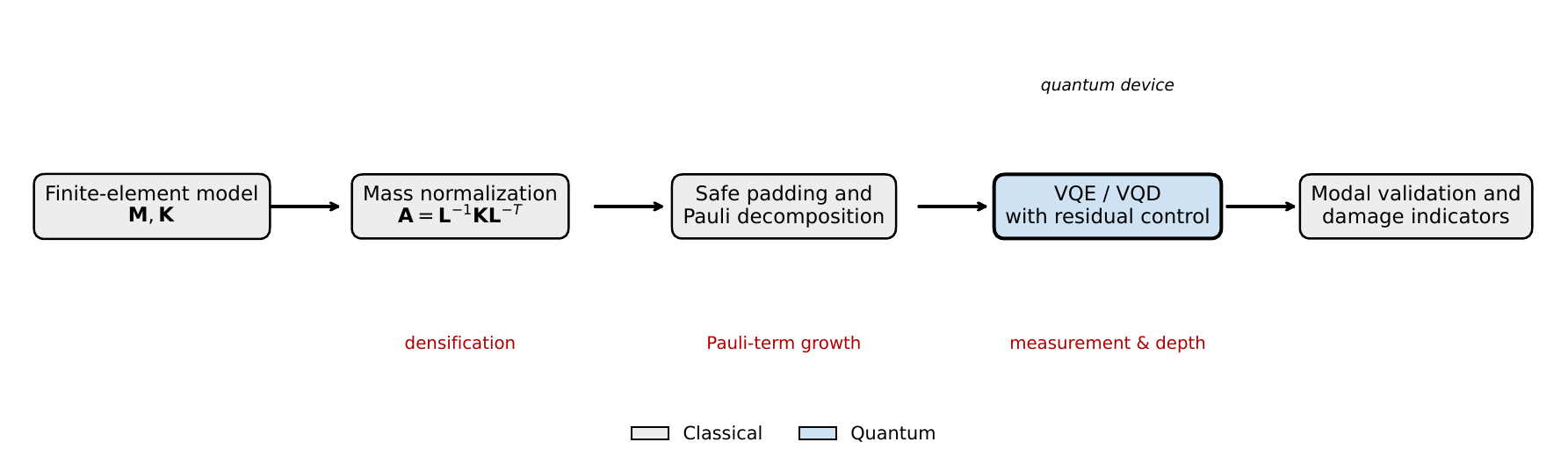}
\caption{End-to-end pipeline. A finite-element model $(\M,\K)$ is mass-normalized by the Cholesky map $\A=\mathbf L^{-1}\K\mathbf L^{-T}$, spectrally padded and Pauli-decomposed, solved by a variational eigensolver (VQE/VQD) with residual control, and finally converted into modal-validation and damage indicators. Each stage is analyzed for correctness, cost, and measurement precision separately. Red labels mark the dominant bottleneck stages; their asymptotic costs are quantified in \cref{tab:complexity}.}
\label{fig:workflow}
\end{figure*}

\section{Structural-dynamics foundation}
\label{sec:structural}

For a linear time-invariant structure,
\begin{equation}
\M\ddot{\uu}(t)+\C\dot{\uu}(t)+\K\uu(t)=\mathbf B_f\ff(t).
\label{eq:eom}
\end{equation}
Under uniform support acceleration $\ddot u_g(t)$,
\begin{equation}
\M\ddot{\uu}+\C\dot{\uu}+\K\uu=-\M\rr\ddot u_g(t).
\label{eq:base}
\end{equation}
For undamped free vibration, substitute $\uu(t)=\ph e^{\mathrm i\omega t}$ into \cref{eq:eom} to obtain \cref{eq:gep}. For a stable restrained structure, $\M$ and $\K$ are symmetric, $\M\succ0$, and the modes may be mass normalized:
\begin{align}
\PhiM_m^T\M\PhiM_m&=\I_m,\label{eq:massorth}\\
\PhiM_m^T\K\PhiM_m&=\boldsymbol\Omega_m^2,
\qquad \boldsymbol\Omega_m=\diag(\omega_1,\ldots,\omega_m).
\label{eq:stifforth}
\end{align}
The modal transformation $\uu=\PhiM_m\mathbf q$ yields, under classical damping,
\begin{equation}
\ddot{\mathbf q}+2\boldsymbol\Xi\boldsymbol\Omega_m\dot{\mathbf q}+\boldsymbol\Omega_m^2\mathbf q
=\PhiM_m^T\mathbf B_f\ff,
\label{eq:modal_eom}
\end{equation}
where $\boldsymbol\Xi=\diag(\xi_1,\ldots,\xi_m)$.

Finite-element assembly gives
\begin{equation}
\K=\sum_{e=1}^{n_e}\mathbf A_e^T\K_e\mathbf A_e,
\qquad
\M=\sum_{e=1}^{n_e}\mathbf A_e^T\M_e\mathbf A_e.
\label{eq:assembly}
\end{equation}
Boundary conditions are imposed before the quantum mapping. Penalty-enforced supports are avoided because artificial high stiffness can worsen conditioning and contaminate the spectrum.

\section{Mass-orthonormal Hermitian mapping and preprocessing}
\label{sec:mapping}

\subsection{Cholesky mapping}

The nonsymmetric product $\M^{-1}\K$ has the correct eigenvalues but is generally not Hermitian in the Euclidean metric. Let
\begin{equation}
\M=\mathbf L\mathbf L^T
\label{eq:chol}
\end{equation}
be a Cholesky factorization and define
\begin{equation}
\yy_i=\mathbf L^T\ph_i,
\qquad
\A=\mathbf L^{-1}\K\mathbf L^{-T}.
\label{eq:map}
\end{equation}
Then
\begin{equation}
\A\yy_i=\lambda_i\yy_i.
\label{eq:standard_eig}
\end{equation}

\begin{proposition}[Mass--Euclidean isometry]
If $\M\succ0$ and $\K=\K^T$, then $\A=\A^T$ and
\begin{equation}
\yy_i^T\yy_j=\ph_i^T\M\ph_j.
\label{eq:isometry}
\end{equation}
Thus Euclidean orthogonality of quantum states is exactly structural mass orthogonality after back-transformation.
\end{proposition}

\begin{proof}
Symmetry follows from $\A^T=\mathbf L^{-1}\K^T\mathbf L^{-T}=\A$. Moreover,
$\yy_i^T\yy_j=\ph_i^T\mathbf L\mathbf L^T\ph_j=\ph_i^T\M\ph_j$.
\end{proof}

The transformation should be implemented with triangular solves, not explicit inverses. Dense Cholesky has $O(N^3)$ arithmetic cost and $O(N^2)$ storage. Sparse cost depends on mesh dimension, ordering, and fill. For a class of two-dimensional finite-element meshes under nested dissection, Cholesky requires $O(N^{3/2})$ operations and $O(N\log N)$ nonzeros \cite{GeorgeNg1988}. These costs are classical preprocessing and must be included in any end-to-end comparison.

\subsection{Whitening can destroy locality}

Sparsity of $\M$ and $\K$ does not imply sparsity of $\A$. When $\M$ is diagonal, whitening is a diagonal scaling and often preserves graph locality. With a consistent mass matrix, $\mathbf L^{-1}$ can introduce long-range couplings. Define
\begin{equation}
\delta_A(\tau)=\frac{1}{N^2}\#\{(i,j):|A_{ij}|>\tau\norm{\A}_{\max}\}.
\label{eq:density}
\end{equation}
The transformed density $\delta_A$ is therefore an algorithmic observable, not a minor implementation detail. Every density reported in this paper uses $\tau=10^{-12}$. This is a different threshold from the $10^{-10}$ coefficient cut used for Pauli counts in \cref{sec:resources}, and the two are not interchangeable: $\tau$ is a \emph{relative} cut on matrix entries, scaled by $\norm{\A}_{\max}$, whereas the Pauli cut is an \emph{absolute} one applied to the coefficients of an operator already normalized to unit spectral radius. Neither implies the other, and \cref{sec:resources} shows the consistent-mass Pauli count to be sensitive to its cut where the density is not.

A generalized Rayleigh quotient can avoid explicit whitening,
\begin{equation}
\mathcal R_G(\theta)=
\frac{\langle\psi(\theta)|\K|\psi(\theta)\rangle}
{\langle\psi(\theta)|\M|\psi(\theta)\rangle},
\label{eq:generalized_rq}
\end{equation}
as in generalized variational formulations \cite{Sato2023}. The trade-off is that both operators must be measured, the denominator is strictly positive because $\M\succ0$ but its usable lower bound is governed by $\kappa(\M)$, and higher-mode orthogonality must be enforced in the $\M$ metric. The Cholesky and generalized-quotient routes should therefore be compared rather than assuming that whitening is free; \cref{sec:genquotient} carries out that comparison at the level of operator representation.

\section{Two-scale spectral normalization and safe padding}
\label{sec:scaling}

\subsection{Two-scale normalization}

A Gershgorin upper bound is
\begin{equation}
U_G=\max_i\left(A_{ii}+\sum_{j\ne i}|A_{ij}|\right)
\ge\lambda_{\max}(\A).
\label{eq:gersh}
\end{equation}
Using $U_G$ as the sole scaling factor is safe but may compress the physical spectrum near zero. We therefore distinguish:
\begin{itemize}
\item an \emph{optimization scale} $s_t>0$, estimated for example by a short power or Lanczos iteration; and
\item the \emph{certified physical bound} $U=U_G/s_t$.
\end{itemize}
The optimization Hamiltonian is
\begin{equation}
\widetilde\A=\frac{\A}{s_t}.
\label{eq:tightscale}
\end{equation}
No assumption that $s_t\ge\lambda_{\max}$ is required.

\subsection{Spectral-safe padding}

Let $n_q=\lceil\log_2N\rceil$ and $D=2^{n_q}$. Define
\begin{equation}
\boxed{
\Hq_p=\begin{bmatrix}
\widetilde\A&\mathbf0\\
\mathbf0&\alpha\I_{D-N}
\end{bmatrix},
\qquad
\alpha=U+\delta_p,
\quad \delta_p>0.
}
\label{eq:safepad}
\end{equation}

\begin{proposition}[Safe padding under arbitrary positive scaling]
Every physical eigenvalue of $\widetilde\A$ is at most $U$, whereas every padded eigenvalue equals $U+\delta_p$. Hence padding cannot enter the physical low spectrum even when $s_t$ underestimates $\lambda_{\max}(\A)$.
\end{proposition}

\subsection{Effect of conservative scaling}

\begin{proposition}[Scalar-scaling law]
Let $\Hq_c=\Hq/c$ with $c>0$ and $E_c(\theta)=\langle\psi(\theta)|\Hq_c|\psi(\theta)\rangle$. Then eigenvectors and condition number are unchanged, while
\begin{equation}
\begin{aligned}
\lambda_i(\Hq_c)&=\lambda_i(\Hq)/c,\\
\nabla E_c&=\nabla E/c,\qquad
\Var_{\psi}(\Hq_c)=\Var_{\psi}(\Hq)/c^2.
\end{aligned}
\label{eq:scalinglaw}
\end{equation}
\end{proposition}

Consequently, excessive scaling does not change the exact eigenvectors, but it reduces absolute gaps and gradients and can interact badly with fixed optimizer tolerances and additive calibration errors. Under ideal shot noise, both signal and standard deviation scale as $1/c$, so relative precision is unchanged if all tolerances and coefficients are rescaled coherently. In practice, coherent rescaling requires the energy stopping tolerance, deflation margins, padding offset, confidence intervals, coefficient-pruning threshold, and any gradient threshold to be transformed by the same factor. The invariance fails when software retains fixed absolute tolerances, when coefficient truncation is not rescaled, or when hardware contributes an additive bias or noise floor that does not scale with the Hamiltonian coefficients. The optimization concern is therefore numerical and hardware-relative, not a universal loss of statistical information.

The physical frequency is recovered from $\widetilde\lambda_i$ by
\begin{equation}
f_i=\frac{1}{2\pi}\sqrt{s_t\widetilde\lambda_i}.
\label{eq:freqrecover}
\end{equation}

\section{Exact deflation conditions for multiple modes}
\label{sec:deflation}

\subsection{Ground-state VQE}

For a normalized parameterized state $|\psi(\theta)\rangle=U(\theta)|0\rangle^{\otimes n_q}$,
\begin{equation}
E(\theta)=\langle\psi(\theta)|\Hq_p|\psi(\theta)\rangle
\label{eq:vqe}
\end{equation}
upper-bounds the lowest eigenvalue accessible within the ansatz manifold. This is a variational statement, not a guarantee that a nonconvex classical optimizer finds the global minimum.

\subsection{Exact deflation theorem}

Let the non-padded eigenpairs of $\Hq_p$ be $(\varepsilon_k,|u_k\rangle)$ with
\begin{equation}
\varepsilon_1\le\varepsilon_2\le\cdots\le\varepsilon_N<\alpha.
\end{equation}
For target mode $r$, define the exact deflated operator
\begin{equation}
\Hq_r=\Hq_p+\sum_{j=1}^{r-1}\beta_jP_j,
\qquad P_j=|u_j\rangle\langle u_j|.
\label{eq:deflatedH}
\end{equation}

\begin{theorem}[Necessary and sufficient deflation threshold]
Assume $\varepsilon_r<\varepsilon_{r+1}$, with $\varepsilon_{N+1}:=\alpha$ when $D>N$ and $\varepsilon_{N+1}:=+\infty$ otherwise. Then $|u_r\rangle$ is the unique ground state of $\Hq_r$ if and only if
\begin{equation}
\boxed{\beta_j>\varepsilon_r-\varepsilon_j\quad\text{for all }j<r.}
\label{eq:betathreshold}
\end{equation}
\end{theorem}

\begin{proof}
The eigenvectors of $\Hq_r$ remain $|u_k\rangle$. For $k<r$, the corresponding eigenvalue is $\varepsilon_k+\beta_k$; for $k\ge r$, it is $\varepsilon_k$; padded states have eigenvalue $\alpha$. The target is the unique minimum precisely when every shifted lower eigenvalue exceeds $\varepsilon_r$ and $\varepsilon_{r+1},\alpha>\varepsilon_r$.
\end{proof}

The sufficiency direction is due to Higgott et al.\ \cite{Higgott2019}, who introduce the deflated cost \cref{eq:deflatedH} and note that $\beta_j>\varepsilon_r-\varepsilon_j$ suffices. What \cref{eq:betathreshold} adds is the converse---no smaller penalty works, so the condition is sharp---and the inclusion of the padded block through the convention $\varepsilon_{N+1}:=\alpha$, which matters because padding contributes eigenvalues that deflation must also clear.

The theorem replaces an arbitrary common penalty with a spectral condition. If a certified target upper bound $U_r\ge\varepsilon_r$ and prior eigenvalue error bounds $|\widehat\varepsilon_j-\varepsilon_j|\le\delta_j$ are available, a deterministic sufficient choice is
\begin{equation}
\boxed{
\beta_j=U_r-\widehat\varepsilon_j+\delta_j+\tau,
\qquad \tau>0.
}
\label{eq:certbeta}
\end{equation}
Indeed, $\varepsilon_j+\beta_j\ge U_r+\tau>\varepsilon_r$. Two bound policies are distinguished. A quantum-only spectral setup may use
\begin{equation}
U_r^{(G)}=U_G/s_t,
\label{eq:globalUr}
\end{equation}
which requires no auxiliary classical eigensolve but can be conservative. Alternatively, if $\vartheta_r^{(p)}$ is the $r$th Ritz value in a $p$-dimensional trial subspace, the min--max principle gives $\varepsilon_r\le\vartheta_r^{(p)}$; hence
\begin{equation}
U_r^{(R)}=\vartheta_r^{(p)}+\delta_{\mathrm{fp}}
\label{eq:ritzUr}
\end{equation}
is a tighter certified bound after allowing for floating-point error. Tighter is not automatically better here: by \cref{eq:betastar} a tighter $U_r$ combined with a small $\tau$ drives the deflated separation $g_r$ toward zero and the Davis--Kahan certificate toward vacuity, whereas the conservative Gershgorin policy \cref{eq:globalUr} overshoots the threshold and lands in the safe regime. A tight bound policy must therefore be paired with a $\tau$ of order $G_r$. A sparse $p$-step Lanczos/Ritz pass costs approximately $O(p\,\mathrm{nnz}(\A)+Np^2+p^3)$ arithmetic ($\mathrm{nnz}(\cdot)$ denoting the number of nonzero entries) and $O(Np)$ storage, and is therefore recorded explicitly as $C_{\mathrm{bound}}$ in the resource vector rather than treated as free preprocessing.

For a previously extracted mode, a practical high-confidence error radius can combine an eigenstate residual and Pauli-sampling uncertainty:
\begin{equation}
\begin{aligned}
\delta_j^{(\alpha_j)}&=\widehat r_j+
\sqrt{2\log(2/\alpha_j)\sum_{\ell}\frac{h_\ell^2}{n_\ell}},\\
\widehat r_j&=\norm{(\Hq_p-\widehat\varepsilon_j\I)|\widehat u_j\rangle}_2.
\end{aligned}
\label{eq:deltabound}
\end{equation}
The second term is a conservative Hoeffding radius for independent Pauli sampling. The residual locates the Rayleigh quotient near some eigenvalue; using it as a bound for mode $j$ additionally requires the estimated state to lie in the isolated modal interval assigned to $j$. In the exact threshold experiment below, exact eigenvalues are used only to verify theorem sharpness, so $\delta_j=0$; operational studies must report whether \cref{eq:globalUr} or \cref{eq:ritzUr} is used and include the corresponding classical cost.

\subsection{Approximate projectors and the cost of oversized penalties}

In practice, previous states are approximate. Let $\widehat P_j$ be the estimated projectors and
\begin{equation}
\widehat\Hq_r=\Hq_p+\sum_{j<r}\beta_j\widehat P_j,
\qquad
\Delta_r=\widehat\Hq_r-\Hq_r.
\label{eq:approxdefl}
\end{equation}
Then
\begin{equation}
\norm{\Delta_r}_2
\le\sum_{j<r}\beta_j\norm{\widehat P_j-P_j}_2.
\label{eq:deflperturb}
\end{equation}
Define the exact deflated separation
\begin{equation}
g_r=\min\left\{
\min_{j<r}(\varepsilon_j+\beta_j-\varepsilon_r),
\varepsilon_{r+1}-\varepsilon_r,
\alpha-\varepsilon_r
\right\}.
\label{eq:deflgap}
\end{equation}
If $\norm{\Delta_r}_2<g_r/2$, Weyl's inequality preserves an isolated target eigenvalue and the Davis--Kahan theorem gives the safe bound \cite{DavisKahan1970}
\begin{equation}
\boxed{
\sin\angle(\widehat u_r,u_r)
\le\frac{2\norm{\Delta_r}_2}{g_r}.
}
\label{eq:dkbound}
\end{equation}
Equation \eqref{eq:deflperturb} is easily misread as an argument for the smallest admissible penalty, since $\norm{\Delta_r}_2$ grows linearly with $\beta_j$. That reading is incomplete, because $\beta_j$ also enters the denominator of \cref{eq:dkbound} through $g_r$. Write
\begin{equation}
G_r=\min(\varepsilon_{r+1}-\varepsilon_r,\ \alpha-\varepsilon_r)
\label{eq:physgap}
\end{equation}
for the separation available from the physical spectrum. Then $g_r$ increases with $\beta_j$ until it saturates at $G_r$, so the certified bound is
\begin{equation}
\boxed{\text{minimized at }\ \beta_j^\star=(\varepsilon_r-\varepsilon_j)+G_r,}
\label{eq:betastar}
\end{equation}
diverges as $\beta_j$ approaches the threshold from above, and grows only linearly beyond $\beta_j^\star$. The smallest certified penalty is therefore the \emph{worst} admissible choice, not the best: each penalty should be raised until the deflated separation saturates the physical gap, and not much further. In \cref{eq:certbeta} this corresponds to sizing the total excess $U_r-\varepsilon_r+\delta_j+\tau$ against $G_r$, rather than driving $\tau$ to zero.

Two caveats attach to \cref{eq:betastar}. It locates the minimum of the \emph{certificate}, not of the error: in the perturbation experiment of \cref{sec:results} the measured target rotation continues to fall as $\beta_j$ is increased past $\beta_j^\star$, so \cref{eq:dkbound} is conservative in a $\beta$-dependent way. And $\beta_j^\star$ requires $\varepsilon_r$ and $G_r$, which are not available before the solve; the practical reading is that a conservative bound policy is safer than a tight one, not that the optimum can be hit exactly.

\subsection{No arbitrary variance or leakage weights}

Hamiltonian variance remains a valuable residual:
\begin{equation}
\sigma_H^2(\psi)=\langle\Hq_p^2\rangle-\langle\Hq_p\rangle^2
=\norm{(\Hq_p-E\I)|\psi\rangle}_2^2.
\label{eq:variance}
\end{equation}
Padding leakage is
\begin{equation}
\ell_p(\psi)=\langle\psi|P_{\mathrm{pad}}|\psi\rangle.
\label{eq:leak}
\end{equation}
Variance has also been studied as an objective in its own right \cite{Zhang2022}; the difficulty here is not variance minimization but the additive combination. Adding $\gamma\sigma_H^2+\zeta\ell_p$ to the energy creates uncalibrated trade-offs and changes the optimization landscape. The primary algorithm therefore uses only
\begin{equation}
C_r(\theta)=\langle\psi(\theta)|\widehat\Hq_r|\psi(\theta)\rangle.
\label{eq:primarycost}
\end{equation}
An optional lexicographic refinement, which preserves the energy ordering within a declared tolerance rather than hiding the trade-off in application-defined weights, is stated in the supplementary material. It is proposed and analyzed but not evaluated experimentally.

\begin{remark}[Degenerate and near-degenerate targets]\label{rem:degenerate}
Theorem~1 assumes $\varepsilon_r<\varepsilon_{r+1}$, and \cref{eq:dkbound} requires $g_r>0$. Neither holds for a degenerate target, and \cref{sec:dim-degeneracy} exhibits a genuine pair with relative gap $2.2\times10^{-16}$ arising from nothing more exotic than a reflection symmetry. Symmetric geometries---square plates, circular tanks, regular framed buildings---make such pairs common rather than exceptional. For a repeated eigenvalue the projector $P_j$ is not basis-independent within the shared eigenspace, so sequential deflation cannot isolate an individual member at all; only the invariant subspace is well defined. Block deflation, replacing $\beta_jP_j$ by $\beta_j P_{\mathcal S_j}$ onto whole eigenspaces, is the natural generalization and is not developed here. This is the same failure mode that motivates the subspace diagnostic of \cref{eq:subspace}.
\end{remark}

\begin{remark}[Scope of the guarantee]\label{rem:scope}
Theorem 1 identifies the correct ground state of the full deflated Hilbert-space operator. It does not prove global convergence of a finite-depth ansatz or a nonconvex classical optimizer. Ansatz approximation error and optimization failure are measured separately through restarts, residuals, and depth studies.
\end{remark}

\subsection{Simultaneous and targeted alternatives}

SSVQE maps orthogonal input states through a common unitary, preserving orthogonality by construction \cite{Nakanishi2019}. For a target shift $\sigma$, the folded operator
\begin{equation}
\Hq_{\sigma}=(\Hq_p-\sigma\I)^2
\label{eq:folded}
\end{equation}
selects the mode nearest $\sigma$. These alternatives are useful when sequential projector error is unacceptable or when a narrow frequency band is known in advance.

\section{Trainability and ansatz expressibility}
\label{sec:trainability}

A real structural Hamiltonian permits a real-amplitude circuit,
\begin{equation}
U_L(\theta)=
\left[\prod_{q=1}^{n_q}R_y(\theta_{qL})\right]
\prod_{\ell=1}^{L}
\left[U_{\mathrm{ent}}^{(\ell)}\prod_{q=1}^{n_q}R_y(\theta_{q,\ell-1})\right].
\label{eq:ansatz}
\end{equation}
This is a real-amplitude specialization of the hardware-efficient family \cite{Kandala2017}, whose original form applies three Euler angles per qubit per layer and therefore spans complex amplitudes; restricting each rotation to $R_Y$ confines the trial state to the real subspace, which suffices because $\Hq$ is real symmetric, and is the restriction whose consequences \cref{sec:trainability} examines. Logarithmic qubit count does not imply shallow representability. A normalized real vector in $D=2^{n_q}$ dimensions has $D-1$ independent amplitudes, whereas a depth-$L$ hardware-efficient circuit has only $O(Ln_q)$ parameters. Therefore, a fixed-depth family cannot be universally expressive as $N$ grows.

Random parameterized circuits may exhibit barren plateaus in which gradient variance decays exponentially with qubit count \cite{McClean2018}. Cost-function locality can materially alter trainability, and even shallow circuits with global costs may exhibit barren plateaus \cite{CerezoBP2021}. Binary amplitude encoding and Cholesky whitening can turn a local structural operator into a nonlocal qubit Hamiltonian, so structural locality does not guarantee a local quantum cost. This concern is directly connected to the scaling results in \cref{sec:results}: for consistent mass, the transformed density grows to 0.308 and the Pauli count to 2950 at $N=128$, indicating that increasingly many computational-basis amplitudes and Pauli supports become coupled. Such densification moves the objective away from the local-cost setting in which trainability can be more favorable and toward the global-cost regime associated with gradient concentration. The present 2--12-qubit data cannot establish a structural barren plateau, but they show that whitening-induced nonlocality can worsen the locality condition as $N$ grows independently of whether circuit depth is increased. Quantifying this interaction for larger qubit counts, alternative orderings, and domain-decomposed mappings is therefore an open problem rather than an assumed asymptotic advantage.

For a single parameter $\theta_a$---here $a=0$, the first $R_y$ angle of the initial rotation layer---the empirical trainability statistic used here is
\begin{equation}
V_g(n_q)=\Var_{\theta\sim\mathcal D}
\left[\frac{\partial E(\theta)}{\partial\theta_a}\right].
\label{eq:gradvar}
\end{equation}
A finite $n_q\le12$ study cannot prove or disprove an asymptotic barren plateau, but it can detect whether an exponential decay has already set in. At the shallowest depth it has not (\cref{sec:results}): the uniform-initialization gradient variance at $L=2$ is better described by a polynomial than by an exponential in $n_q$, and its high-qubit slope is far shallower than an exponential collapse. That favourable reading is specific to that depth, and \cref{sec:results} shows it reversing as the depth is raised toward the parameter-count threshold for universal expressibility.

\section{Recovery, structural identification, and health-monitoring observables}
\label{sec:observables}

\subsection{Physical recovery and validation vector}

After removing padded coordinates, recover and mass-normalize
\begin{equation}
\ph_i=\mathbf L^{-T}\yy_i,
\qquad
\ph_i\leftarrow\frac{\ph_i}{\sqrt{\ph_i^T\M\ph_i}}.
\label{eq:recover}
\end{equation}
The normalized equilibrium residual is
\begin{equation}
\rho_i=
\frac{\norm{\K\ph_i^Q-\lambda_i^Q\M\ph_i^Q}_2}
{\norm{\K\ph_i^Q}_2+|\lambda_i^Q|\norm{\M\ph_i^Q}_2}.
\label{eq:residual}
\end{equation}
Multi-mode errors are
\begin{align}
\epsilon_M&=\norm{(\PhiM^Q)^T\M\PhiM^Q-\I}_F,\label{eq:em}\\
\epsilon_K&=\frac{\norm{(\PhiM^Q)^T\K\PhiM^Q-(\boldsymbol\Omega^Q)^2}_F}
{\norm{(\boldsymbol\Omega^Q)^2}_F}.\label{eq:ek}
\end{align}
No scalar exponential integrity index is introduced. Instead, every mode is reported through the non-aggregated diagnostic vector
\begin{equation}
\boxed{
\mathcal D_i=
(\epsilon_{f,i},\rho_i,\sigma_{H,i}^2,\ell_{p,i},o_i,1-\MAC_{ii}^{(M)}),
}
\label{eq:diagnosticvector}
\end{equation}
where $o_i=\max_{j<i}|\langle\widehat u_j|\widehat u_i\rangle|^2$. Acceptance thresholds are application-specific and must be stated component by component.

The mass-weighted modal assurance criterion is \cite{Allemang1982}
\begin{equation}
\MAC_{ij}^{(M)}=
\frac{|(\ph_i^C)^T\M\ph_j^Q|^2}
{[(\ph_i^C)^T\M\ph_i^C][(\ph_j^Q)^T\M\ph_j^Q]}.
\label{eq:mac}
\end{equation}
For a near-degenerate group, compare invariant subspaces. If
\begin{equation}
(\mathbf Y_C)^T\mathbf Y_Q=\mathbf U\cos\boldsymbol\Theta\mathbf V^T,
\end{equation}
then
\begin{equation}
\epsilon_{\mathrm{sub}}=\norm{\sin\boldsymbol\Theta}_F
\label{eq:subspace}
\end{equation}
is invariant to rotations inside the repeated eigenspace.

\subsection{Participation and frequency response}

For base excitation,
\begin{equation}
\Gamma_i=\frac{\ph_i^T\M\rr}{\ph_i^T\M\ph_i},
\qquad M_i^*=\Gamma_i^2
\label{eq:participation}
\end{equation}
for mass-normalized modes. The truncated modal frequency-response function is
\begin{equation}
\Hq_m(\omega)=\sum_{i=1}^{m}
\frac{\ph_i\ph_i^T}
{\omega_i^2-\omega^2+2\mathrm i\xi_i\omega_i\omega}.
\label{eq:frf}
\end{equation}
Frequency accuracy without participation, shape, and response accuracy is not sufficient for structural use.

\subsection{Connection to output-only identification}

In an operational setting the comparison modes come from measurement rather than from a model. Frequency-domain decomposition takes the dominant singular vector of the output power spectral density near resonance as the measured mode shape \cite{Brincker2001}, while stochastic subspace identification extracts poles from a discrete state-space realization, giving frequencies and damping ratios from the continuous poles $s_i=\log\mu_i/\Delta t$ \cite{Peeters1999,Peeters2001,Reynders2012}. Experimental covariance estimates can weight the modal assurance criterion accordingly \cite{Reynders2008}. Two consequences matter here: identified shapes exist only at instrumented coordinates, and they carry sampling uncertainty that any comparison metric must absorb.

\subsection{Unbiased temporal mode tracking}

Repeated monitoring epochs are solved independently and mode labels assigned afterwards, so that continuity cannot pull a damaged mode toward its healthy predecessor through the eigensolver objective. The assignment rule, its uncertainty normalization, and its limitations are given in the supplementary material; like the lexicographic refinement it is proposed and analyzed but not evaluated experimentally.

\subsection{Sensor-space recovery}

For sensor selector $\mathbf c_s$,
\begin{equation}
\phi_{i,s}=\mathbf c_s^T\ph_i
=\mathbf c_s^T\mathbf L^{-T}\yy_i
=\mathbf d_s^T\yy_i,
\qquad \mathbf d_s=\mathbf L^{-1}\mathbf c_s.
\label{eq:sensor}
\end{equation}
Therefore
\begin{equation}
\phi_{i,s}=\norm{\mathbf d_s}_2\langle\widehat d_s|\psi_i\rangle.
\label{eq:sensoroverlap}
\end{equation}
This can reduce readout when only instrumented coordinates are needed, although preparing $|\widehat d_s\rangle$ has its own cost. It is resource-favorable only when $\mathbf d_s=\mathbf L^{-1}\mathbf c_s$ has exploitable sparsity, low-rank structure, or a reusable preparation oracle; if whitening makes $\mathbf d_s$ dense and unstructured, generic amplitude preparation can require $O(N)$ gates and may erase the advantage over recovering the required classical amplitudes directly.

\subsection{Damage sensitivity and modal strain energy}

Let element damage reduce stiffness:
\begin{equation}
\K(\mathbf d)=\sum_{e=1}^{n_e}(1-d_e)\mathbf A_e^T\K_e^0\mathbf A_e.
\label{eq:damageK}
\end{equation}
For mass-normalized modes and stiffness-only damage,
\begin{equation}
\frac{\partial\lambda_i}{\partial d_e}
=-\ph_i^T\K_e^0\ph_i.
\label{eq:sensitivity}
\end{equation}
The element modal-strain-energy fraction is formed with the \emph{baseline} element stiffness $\K_e^0$ in the numerator, for both the healthy and the damaged structure,
\begin{equation}
\eta_{e,i}=
\frac{\ph_i^T\K_e^0\ph_i}{\ph_i^T\K\ph_i},
\label{eq:mse}
\end{equation}
in the spirit of standard modal-strain-energy practice \cite{Stubbs1995,Cornwell1999}, though not identical to it: those formulations hold the undamaged stiffness fixed in both numerator and denominator and compare the resulting fractions as a ratio, whereas \cref{eq:mse} retains the current global $\K$ in the denominator---where it is the measured eigenvalue---and \cref{eq:damageindex} takes an absolute difference. The absolute difference is biased toward elements carrying high modal strain energy, which a ratio formulation would partly normalize away. This is not a convention of convenience: $\K_e(\mathbf d)$ is the unknown being inferred, so an operator built from it could not be assembled, let alone measured. The denominator is unaffected, since $\ph_i^T\K\ph_i=\lambda_i$ for mass-normalized modes and is therefore supplied by the eigensolver itself. With
\begin{equation}
\Hq_e=\frac{\mathbf L^{-1}\K_e^0\mathbf L^{-T}}{s_t},
\label{eq:He}
\end{equation}
one obtains
\begin{equation}
\eta_{e,i}^{(\epoch)}=
\frac{\langle\psi_i^{(\epoch)}|\Hq_e|\psi_i^{(\epoch)}\rangle}
{\langle\psi_i^{(\epoch)}|\widetilde\A^{(\epoch)}|\psi_i^{(\epoch)}\rangle},
\label{eq:qmse}
\end{equation}
Two caveats bound the practical reach of \cref{eq:qmse}. The numerator operator is built from the baseline model and is therefore available, but the state $|\psi_i\rangle$ is not. Preparing the damaged mode requires either an eigensolve of the damaged operator, which presumes exactly the model being inferred, or amplitude encoding of a mode shape identified classically from measured response---in which case $\ph_i^T\K_e^0\ph_i$ is an $O(\mathrm{nnz}(\K_e))$ classical inner product and no measurement strategy can make the quantum evaluation competitive. Operational modal analysis moreover returns shapes only at instrumented coordinates, so the full $N$-vector that amplitude encoding needs does not exist. The finite-shot study of \cref{sec:results} accordingly isolates measurement variance under a \emph{known} damaged model; it quantifies the cost of the expectation-value formulation and is not a demonstration of a deployable monitoring method.

A multimode change score is
\begin{equation}
D_e=\sum_{i=1}^{m}w_i|\eta_{e,i}^{d}-\eta_{e,i}^{0}|,
\qquad \sum_iw_i=1.
\label{eq:damageindex}
\end{equation}
This class of energy-based localization is related to established modal strain-energy methods \cite{Stubbs1995,Cornwell1999}. The finite-shot study below tests whether the quantum expectation formulation remains usable when the exact modes are no longer assumed available.

\section{Measurement and resource accounting}
\label{sec:resources}

\subsection{Pauli representation}

For $n_q$ qubits,
\begin{equation}
\Hq=\sum_{\ell=1}^{N_P}h_\ell P_\ell,
\qquad P_\ell\in\{I,X,Y,Z\}^{\otimes n_q},
\label{eq:pauli}
\end{equation}
where
\begin{equation}
h_\ell=2^{-n_q}\tr(P_\ell\Hq).
\end{equation}
A real symmetric $\Hq$ has zero coefficient on every Pauli string containing an odd number of $Y$ factors, by the parity argument used in the proof of \cref{prop:lumped}, so at most $2^{2n_q-1}+2^{n_q-1}\le2N^2$ of the $4^{n_q}$ strings can be nonzero. The worst-case Pauli basis therefore still grows quadratically in the encoded structural dimension, at half the naive rate. Several consistent-mass cases below \emph{saturate} this bound exactly---136 at $n_q=4$ and 2080 at $n_q=6$---so it is attained, not merely an upper limit. Actual nonzero count depends on ordering, padding, mass model, and whitening. Commuting-group methods can reduce circuit settings but not eliminate finite-sampling error \cite{Crawford2021,Gokhale2020}; on molecular Hamiltonians such partitioning reduces the measurement count from $O(N^4)$ to $O(N^3)$ \cite{Gokhale2020}, a reduction in distinct settings rather than in sampling variance, exactly the distinction drawn in \cref{sec:genquotient}. Throughout, a coefficient is counted as nonzero when $|h_\ell|>10^{-10}$. That absolute cut is not innocuous for the consistent-mass operator, whose coefficient magnitudes show no numerical-zero plateau: at $N=128$ the count is 2316, 2950, 3640 and 4251 for cuts $10^{-8}$, $10^{-10}$, $10^{-12}$ and $10^{-14}$, and the fitted exponent of \cref{sec:complexity} moves from 1.36 to 1.65 over the same range. Lumped-mass counts are insensitive to the cut. Reported Pauli counts should therefore be read together with the threshold that produced them.

For a diagonal mass matrix the quadratic bound is far from attained, and the exact count is available in closed form.

\begin{proposition}[Lumped-mass Pauli count]
\label{prop:lumped}
Let the mesh be a $d$-dimensional grid with axis lengths $L_1,\ldots,L_d$, each a power of two, indexed row-major so that axis $a$ carries stride $\sigma_a=\prod_{b>a}L_b$, with nearest-neighbour couplings along each axis, and let $\M$ be diagonal so that $\A=\M^{-1/2}\K\M^{-1/2}$ inherits the sparsity pattern of $\K$. Then
\begin{equation}
N_P\le(n_q+2)\,2^{\,n_q-1}=O(N\log N),
\label{eq:lumpedcount}
\end{equation}
independently of the spatial dimension $d$, with equality if and only if
\begin{equation}
S(x,z)=\sum_i(-1)^{|i\wedge z|}A_{i,i\oplus x}\neq0
\label{eq:lumpedequality}
\end{equation}
for every admissible pair $(x,z)$: that is, for every mask $x$ that is either zero or one of the $n_q$ nonzero coupling masks, together with every $z$ such that $|x\wedge z|$ is even.
\end{proposition}

\begin{proof}
Write a Pauli string as $P=\mathrm i^{n_Y}X^{x}Z^{z}$ with $n_Y=|x\wedge z|$. Then $P$ connects only the index pairs $(i,i\oplus x)$, so a mask $x$ contributes nothing unless $A_{i,i\oplus x}\ne0$ for some $i$. Along axis $a$ the neighbour pairs are $(i,i+\sigma_a)$, and for a power-of-two axis length $L_a$ the values $i\oplus(i+\sigma_a)$ realise exactly $\log_2 L_a$ distinct masks; the strides occupy disjoint bit ranges, so masks from different axes are distinct. The nonzero off-diagonal masks therefore number $\sum_a\log_2 L_a=\log_2 N=n_q$, with $x=0$ contributing the diagonal.

Fix $x$ and let $S(x,z)=\sum_i(-1)^{|i\wedge z|}A_{i,i\oplus x}$, so that the coefficient is $2^{-n_q}\mathrm i^{n_Y}S$. Re-indexing $i\mapsto i\oplus x$ and using $\A=\A^T$ gives $S=(-1)^{|x\wedge z|}S$, so $S$ vanishes whenever $n_Y$ is odd. For $x\ne0$ exactly half of the $2^{n_q}$ masks $z$ satisfy $|x\wedge z|$ even, leaving $2^{n_q-1}$ surviving coefficients per nonzero mask, while $x=0$ admits all $2^{n_q}$. Hence $N_P$ equals the number of pairs $(x,z)$ with $|x\wedge z|$ even and $S(x,z)\ne0$, which is at most $2^{n_q}+n_q2^{n_q-1}$, with equality exactly under \cref{eq:lumpedequality}.
\end{proof}

The $x=0$ sector is not automatic and must be included. There $S(0,z)$ is the Walsh--Hadamard
transform of the diagonal of $\A$, so a constant diagonal gives $S(0,z)=0$ for every $z\neq0$ and the
bound is missed by $2^{n_q}-1$ terms regardless of the off-diagonal conditions. Heterogeneity of the
diagonal is therefore as necessary for equality as heterogeneity of the coupling diagonals, and it is
the dominant mechanism behind the uniform-chain deficit reported below: at $n_q=3,4,5$ the uniform
chain retains only one surviving pair in the $x=0$ sector, against the full $2^{n_q}$ retained by the
heterogeneous chain.

Condition \cref{eq:lumpedequality} is checkable and not vacuous. For fixed $x$, $S(x,\cdot)$ is the Walsh--Hadamard transform of the vector $v^{(x)}_i=A_{i,i\oplus x}$ carried on that coupling diagonal, so equality asks only that no transform coefficient vanish on the even-parity indices. Heterogeneous properties make this generic; constant ones destroy it. On a homogeneous mesh both the diagonal and each $v^{(x)}$ are constant across their support, their transforms concentrate on a vanishing fraction of $z$, and a uniform lumped chain carries only $3\cdot2^{\,n_q-1}-1$ nonzero terms ($5,11,23,47,95$ for $n_q=2,\ldots,6$), i.e.\ $\Theta(N)$ rather than $\Theta(N\log N)$. Evaluating \cref{eq:lumpedequality} directly reproduces both regimes exactly: 20, 48 and 112 surviving pairs for the heterogeneous chains at $n_q=3,4,5$, against 11, 23 and 47 for their uniform counterparts. The heterogeneous meshes used throughout this paper are chosen precisely to avoid such cancellations, so \cref{eq:lumpedcount} is attained there while a uniform benchmark would report fewer terms. The count is also ordering-dependent: it characterizes the lumped operator \emph{in a locality-preserving indexing}, and a scrambled DOF order moves it back toward the $O(N^2)$ worst case. \Cref{sec:ordering} measures this, and finds that a scramble saturates the worst case exactly, while even a standard bandwidth-reducing reordering costs an order of magnitude in $N_P$, and does so at no fill benefit at all under lumped mass.

\Cref{eq:lumpedcount} matches every lumped-mass count reported below exactly, in one, two, and three dimensions: $8,20,48,112,256,576$ for $n_q=2,\ldots,7$, and $1280$ and $6144$ at $n_q=8$ and $10$ (\cref{tab:scaling,tab:dimscaling}). It also explains why those counts coincide across spatial dimensions at matched $N$: the number of distinct coupling masks is $\log_2 N$ however the axes are partitioned. The measurement burden of the lumped model is therefore at most $(n_q+2)2^{\,n_q-1}=O(N\log N)$, not $O(N^2)$, and is attained by the heterogeneous meshes studied here; within this locality-preserving nearest-neighbour family the quadratic regime is reached through the mass model rather than through dimension or size, which is the mechanism quantified in \cref{sec:results}; a scrambled ordering or non-nearest-neighbour coupling would reach it by another route.

Arbitrary amplitude-state preparation may itself require gate counts exponential in qubit number in the worst case, despite logarithmic qubit storage \cite{PleschBrukner2011}. Resource reports must therefore include state preparation, two-qubit gates, circuit depth, Pauli terms, commuting groups, shots, and optimizer evaluations.

\subsection{Finite-shot propagation}

For an independently measured Pauli mean $\mu_\ell$ with $n_\ell$ shots,
\begin{equation}
\Var(\widehat\mu_\ell)=\frac{1-\mu_\ell^2}{n_\ell}.
\label{eq:paulivar}
\end{equation}
Ignoring covariance between groups,
\begin{equation}
\Var(\widehat E)=\sum_\ell h_\ell^2\frac{1-\mu_\ell^2}{n_\ell}.
\label{eq:energyvar}
\end{equation}
For fixed total shots, the continuous optimal allocation satisfies
\begin{equation}
n_\ell\propto |h_\ell|\sqrt{1-\mu_\ell^2}.
\label{eq:allocation}
\end{equation}
The delta method gives
\begin{equation}
\Var(\widehat f_i)\approx
\left(\frac{s_t}{8\pi^2f_i}\right)^2\Var(\widehat\varepsilon_i).
\label{eq:freqvar}
\end{equation}
The $1/f_i^2$ factor makes the fundamental frequency especially vulnerable when its eigenvalue is small compared with the total Pauli-estimator variance.

\subsection{Controlled error budget}

An additive scalar ``total error'' is generally not identifiable because the components interact. We report the controlled error vector
\begin{equation}
\mathcal E=
(\epsilon_{\mathrm{FE}},\epsilon_{\mathrm{map}},\epsilon_{\mathrm{ansatz}},
\epsilon_{\mathrm{opt}},\epsilon_{\mathrm{shot}},\epsilon_{\mathrm{channel}},
\epsilon_{\mathrm{readout}}).
\label{eq:errorvector}
\end{equation}
In this paper, mapping error is zero up to floating-point arithmetic; ansatz and optimizer error are separated using exact state vectors and multiple restarts; shot error is evaluated by Monte Carlo sampling; and a stylized channel/readout sensitivity is reported only as a diagnostic. Finite-element modelling error and real-device error remain unquantified and are explicitly excluded from claims.

No dimensionally mixed ``information efficiency'' is defined. Instead, the raw resource vector is reported:
\begin{equation}
\begin{aligned}
\mathcal R=(&n_q,N_P,N_G,N_{\mathrm{shot}},N_{\mathrm{eval}},D_{\mathrm{circ}},N_{2q},\\
&C_{\mathrm{asm}},C_{\mathrm{chol}},C_{\mathrm{bound}},C_{\mathrm{Pauli}},C_{\mathrm{prep}}).
\end{aligned}
\label{eq:resourcevector}
\end{equation}
where $N_G$ is the number of measurement groups; $C_{\mathrm{asm}}$, $C_{\mathrm{chol}}$, $C_{\mathrm{bound}}$, and $C_{\mathrm{Pauli}}$ denote classical matrix assembly, whitening, spectral-bound, and Pauli-construction costs; and $C_{\mathrm{prep}}$ is quantum state-preparation cost. When the global Gershgorin policy \cref{eq:globalUr} is used, $C_{\mathrm{bound}}$ contains only the row-sum calculation. When a Ritz/Lanczos bound is used, its iterations, matrix--vector products, storage, and projected eigensolve are reported explicitly.

\subsection{End-to-end asymptotic cost}
\label{sec:complexity}

The resource vector \cref{eq:resourcevector} collapses into asymptotic order for the tested one-dimensional chains, where $\M$ and $\K$ are banded with constant half-bandwidth $b$. In the following, $k$ is the number of requested modes, $N_P$ the number of nonzero Pauli terms \eqref{eq:pauli}, $N_{\mathrm{eval}}$ the number of optimizer evaluations, $T_{\mathrm{it}}$ the classical Lanczos iteration count, and $\varepsilon$ the target accuracy of a single energy estimate. \Cref{tab:complexity} compares a sparse classical eigensolver with the variational quantum pipeline for the $k$ lowest modes of an $N$-DOF chain.

\begin{table}[t]
\centering
\caption{End-to-end asymptotic time and space for the $k$ lowest modes of an $N$-DOF banded 1D chain (constant half-bandwidth $b$). Dashes mark steps the classical solver does not perform.$^{\ddagger}$}
\label{tab:complexity}
\begin{tabular}{lll}
\toprule
Stage & Classical (Lanczos) & Quantum (VQD)\\
\midrule
Assembly & $O(N)$ & $O(N)$\\
Cholesky of $\M$ & $O(Nb^2)^{\P}$ & $O(Nb^2)=O(N)$\\
Whitening $\A$ & --- & $O(N^2)$\\
Pauli decomposition & --- & $O(N^2\log N)$\\
Generic real-state expressibility$^{\S}$ & --- & $\Omega(N/\log N)$ depth\\
Solve / measure & $O(kN\,T_{\mathrm{it}})$ & $O(kN^2\varepsilon^{-2}N_{\mathrm{eval}})$\\
\midrule
Total time & $O(kN\,T_{\mathrm{it}})$ & $O(kN^2\varepsilon^{-2}N_{\mathrm{eval}})$\\
Space & $O(Nk)$ & $O(N)+O(N^2)^{\dagger}$\\
\bottomrule
\end{tabular}

\vspace{0.35em}
\parbox{\columnwidth}{\footnotesize $T_{\mathrm{it}}$ is the Lanczos iteration count: roughly constant with shift-invert ARPACK, $O(\sqrt N)$--$O(N)$ unpreconditioned. $^{\dagger}$$O(N)$ statevector amplitudes plus $O(N^2)$ storage for the whitened operator and its Pauli list. $^{\P}$A generalized Lanczos also factorizes $\M$; negligible for banded $\M$. $^{\S}$Every circuit execution must first prepare the trial state. \Cref{rem:depth} bounds below by $L=\Omega(N/\log_2 N)$ the depth at which this real-amplitude ansatz becomes \emph{generically} expressive over real states, by parameter counting; that is not a lower bound for any individual structural eigenstate. Arbitrary amplitude-state preparation is separately exponential in qubit number in the worst case, i.e.\ polynomial in $N$. The quantum column therefore counts circuit \emph{executions} and the classical column floating-point operations; for generic-state scaling within this ansatz, the depth bound is an additional gate-level resource proxy, while no corresponding depth lower bound is established here for an individual structural eigenstate. $^{\ddagger}$Padding and the Gershgorin bound cost $O(N)$ and $O(N^2)$, both dominated by whitening. The tighter Ritz policy \eqref{eq:ritzUr} instead costs $O(p\,\mathrm{nnz}(\A)+Np^2+p^3)$ time and $O(Np)$ space, reported separately as $C_{\mathrm{bound}}$. The whitening and Pauli rows assume band-aware triangular solves and a fast Walsh--Hadamard projection; the reference implementation uses dense LAPACK and an explicit term-by-term projection, so its wall time is not the tabulated asymptote.}
\end{table}

Three consequences follow. First, preprocessing is not the bottleneck for these chains: the banded Cholesky is $O(N)$, not the $O(N^3)$ dense cost noted in \cref{sec:mapping}. Second, the classical solver acts on the sparse $\K$ directly, so its cost scales with $\mathrm{nnz}(\K)=O(N)$, whereas whitening densifies $\A$---from density $0.023$ to $0.308$ between lumped and consistent mass at $N=128$---so the Pauli representation carries $N_P=O(N^2)$ terms in the worst case ($2^{2n_q-1}+2^{n_q-1}\le2N^2$). This quadratic figure is an upper bound: over the tested range the consistent-mass density and Pauli count grow sub-quadratically (empirical exponent ${\approx}1.4$--$1.5$ fitted over $N\ge16$, \cref{fig:pauliscale}), and no universal quadratic law is claimed. The local exponent moreover falls with $N$ over $N\ge8$, reaching 1.25 for the last tested doubling, so the growth is decelerating rather than settling onto a power law. Even so, merely constructing and storing that operator is superlinear in both time and space, and already exceeds the entire classical eigensolve once $N$ is large relative to $k$, before a single shot is taken. Third, the measurement phase $O(kN^2\varepsilon^{-2}N_{\mathrm{eval}})$ dominates the quantum column; by \cref{eq:freqvar} the low-frequency modes are the most sampling-intensive because the eigenvalue-to-frequency map amplifies variance as $1/f_i^2$. That expression treats $\varepsilon$ as an absolute energy tolerance, which understates the growth: the decision-relevant requirement is a fixed \emph{relative} accuracy on a frequency. With $\Hq$ normalized so $\norm{\Hq}\approx1$, $\Var(\widehat E)$ is $O(1/n)$ independently of $N$, whereas $\lambda_1/\lambda_{\max}=O(N^{-2})$ for a chain; holding the fundamental-frequency error fixed therefore demands $\varepsilon=O(N^{-2})$, hence $\varepsilon^{-2}=O(N^{4})$ and a measurement term of $O(kN^{6}N_{\mathrm{eval}})$ for the whitened route; \cref{sec:genquotient} shows that this representation factor is formulation-dependent. The mechanism is already visible at $N=6$ in \cref{tab:shotfreq}: the mode-1 to mode-4 RMSE ratio is 31.9 at $10^5$ shots per term, against $(f_4/f_1)^2=31.6$. The classical cost is left as $O(kN\,T_{\mathrm{it}})$ rather than a single power of $N$ because $T_{\mathrm{it}}$ depends on the solver and the spectral gap; under either policy it is far below the quantum column. The optimizer factor $N_{\mathrm{eval}}$ and the shots needed per gradient enter only through trainability, and \cref{sec:results} finds the gradient variance decaying polynomially rather than exponentially in $n_q$ at fixed shallow depth, i.e.\ a polylogarithmic factor in $N$---though that favourable reading does not survive increasing the depth toward expressibility. The no-advantage conclusion rests on the measurement bottleneck either way, and so does not depend on assuming an untrainable barren plateau.

\subsection{Does avoiding whitening reduce the measurement cost?}
\label{sec:genquotient}

The generalized quotient \eqref{eq:generalized_rq} avoids forming $\A$, replacing one measured
operator by two. Whether that trades down depends on how much whitening densified the operator
to begin with. \Cref{tab:genquotient} settles the representation question: $\K$, $\M$ and $\A$
decomposed under the identical coefficient cut $10^{-10}$, the identical locality-preserving
indexing, and the same greedy grouping used throughout, each operator normalized by its own
largest eigenvalue so the absolute cut means the same thing for all three. The $N_P(\A)$ and
$G_\A$ columns reproduce \cref{tab:scaling} exactly. The scope is deliberately narrow: this is a
comparison of Pauli-term and commuting-group \emph{representation} cost only. No end-to-end
generalized-eigenvalue solver is implemented, the finite-shot bias and variance of the ratio
estimator are not measured, and no $\M$-metric deflation threshold or its measurement cost is
derived. What follows therefore establishes that whitening can be avoided at the representation
level; it does not establish a total runtime for the generalized route.

\begin{table}[t]
\centering
\caption{Operator-representation cost of the generalized quotient \eqref{eq:generalized_rq},
which measures $\K$ and $\M$, against the whitened route, which measures $\A$.
$\Sigma=N_P(\K)+N_P(\M)$ and $\rho=\Sigma/N_P(\A)$; $G$ counts greedy qubit-wise commuting
groups, $G_{\K\M}$ those of the pooled $\K$ and $\M$ term list. Coefficient cut $10^{-10}$
throughout. Dashes mark groupings not evaluated at $N=128$.}
\label{tab:genquotient}
\small
\setlength{\tabcolsep}{4pt}
\begin{tabular}{rrrrrrrr}
\toprule
$N$ & $N_P(\K)$ & $N_P(\M)$ & $\Sigma$ & $N_P(\A)$ & $\rho$ & $G_{\K\M}$ & $G_\A$\\
\midrule
\multicolumn{8}{l}{\emph{Lumped mass}}\\
4   & 8   & 4   & 12  & 8    & 1.50 & 4   & 4\\
8   & 19  & 8   & 27  & 20   & 1.35 & 8   & 8\\
16  & 42  & 16  & 58  & 48   & 1.21 & 16  & 16\\
32  & 89  & 32  & 121 & 112  & 1.08 & 32  & 32\\
64  & 184 & 64  & 248 & 256  & 0.97 & 64  & 64\\
128 & 375 & 128 & 503 & 576  & 0.87 & --- & ---\\
\midrule
\multicolumn{8}{l}{\emph{Consistent mass}}\\
4   & 8   & 8   & 16  & 10   & 1.60 & 4   & 5\\
8   & 19  & 19  & 38  & 36   & 1.06 & 8   & 14\\
16  & 42  & 42  & 84  & 136  & 0.62 & 16  & 41\\
32  & 89  & 89  & 178 & 459  & 0.39 & 32  & 101\\
64  & 184 & 184 & 368 & 1241 & 0.30 & 64  & 221\\
128 & 375 & 375 & 750 & 2950 & 0.25 & --- & ---\\
\bottomrule
\end{tabular}
\end{table}

Avoiding whitening buys back precisely what whitening costs, and only that. Under lumped mass,
\cref{prop:lumped} already guarantees that $\A$ inherits the sparsity of $\K$, so there is
nothing to recover: the generalized route needs 503 Pauli terms against 576 at $N=128$, a
13\% saving confined to the largest two sizes---the ratio exceeds one for $N\le32$---and
exactly the same number of circuit settings at every size tested. Under consistent mass, where
whitening is what destroys locality, the saving is real and grows with $N$: 750 terms against
2950 at $N=128$ (a factor of 3.9), and 64 settings against 221 at $N=64$ (a factor of 3.5).
The unwhitened pair stays banded at every size, and both are real symmetric with
nearest-neighbour couplings in the same locality-preserving indexing, so the counting argument
of \cref{prop:lumped} bounds each by $(n_q+2)2^{\,n_q-1}=O(N\log N)$---the regime whitening
forfeits under consistent mass.

The settings columns add a point the term counts miss. Pooling the two term lists never opens
a group beyond those $\K$ already requires: every $\M$ string is absorbed into an existing
$\K$ group, in all ten cases and under either concatenation order, because the two operators
share a sparsity pattern and hence a set of qubit-wise commuting classes. Measuring a second
operator therefore costs no additional circuit setting here, though the greedy count is a
heuristic upper bound and a better grouping of $\A$ would narrow the gap.

Four costs remain, none of them measured. First, the denominator must itself be estimated, so
$\mathcal R_G$ is a ratio of two noisy expectations: biased at finite shots, with a variance
combining both, and with no counterpart in the whitened formulation where the denominator is
identically one. Second, conditioning. The denominator is strictly positive for every normalized trial state
because $\M$ is positive definite; after normalization by $\lambda_{\max}(\M)$ its certified
lower bound is $\lambda_{\min}(\M)/\lambda_{\max}(\M)=1/\kappa(\M)$, so a vanishing denominator
is excluded by the structural assumptions. Poor conditioning can nonetheless make it small and
amplify the bias and variance of a ratio estimator. For the chain family studied here that risk
is mild: $\kappa(\M)\le2.50$ for lumped and $\le4.92$ for consistent mass across $N\le128$, both
plateauing, which bounds the normalized denominator below by about $0.2$. Third,
orthogonality must be enforced in the $\M$ metric, so each penalty term in \eqref{eq:deflatedH}
becomes the metric overlap $|\langle\psi|\M|\psi_j\rangle|^2$, a transition matrix element
requiring an overlap/interference measurement protocol rather than an ordinary single-state expectation, and the thresholds of \eqref{eq:betathreshold} would need restating in
that metric. Fourth, both operators must reach matched precision, so the shot budget is split
rather than inherited. A representation saving is necessary for the generalized route to help,
not sufficient.

Avoiding whitening changes more than the constant factor. Under the locality-preserving
indexing used here both $\K$ and $\M$ admit $O(N\log N)$ Pauli representations by the counting
argument of \cref{prop:lumped}, whereas the whitened consistent-mass operator can approach the
$O(N^2)$ real-symmetric limit; over the tested range the fitted exponents are $1.05$ for the
unwhitened pair against $1.48$ for $\A$ (both for $N\ge16$), and the ratio in
\cref{tab:genquotient} falls monotonically rather than settling on a constant. If the
fixed-relative-accuracy requirement continues to impose the $\varepsilon^{-2}=O(N^4)$ factor of
\cref{sec:complexity}, the representation-plus-independent-sampling baseline therefore changes
from $O(kN^6N_{\mathrm{eval}})$ for the whitened worst case to $O(kN^5\log N\,N_{\mathrm{eval}})$
for the unwhitened pair. That is an asymptotic improvement in the representation factor, not a
constant-factor one.

Two qualifications are essential. First, $O(kN^5\log N\,N_{\mathrm{eval}})$ is a
representation-plus-shot baseline and not an end-to-end generalized-quotient complexity: it
excludes all four costs listed above, none of which is measured here, and a complete accounting
requires an end-to-end solver together with a covariance-aware shot analysis. Second, the
practical conclusion does not change. $O(kN^5\log N\,N_{\mathrm{eval}})$ remains far above the
$O(kN\,T_{\mathrm{it}})$ of a sparse classical eigensolver for the problems tested here, so
avoiding whitening improves the representation factor without removing the decision-precision
bottleneck or establishing any advantage.

\section{Numerical program}
\label{sec:numerical}

The numerical program is deliberately idealized. Exact circuit states are used in place of noisy prepared states, Pauli expectations are sampled from exact means with independent binomial noise, and no gate, readout, or decoherence error is modelled. This is not an oversight but the appropriate design for a resource lower bound: every one of these idealizations is favourable to the quantum method, so a measurement cost that is already prohibitive under them remains prohibitive under any realistic error model. The corollary is that results in this paper which are \emph{positive}---the depth-3 exact-state recovery of \cref{tab:vqe_depth} and the finite-shot localization accuracy of \cref{sec:results}---are upper bounds on achievable performance and must not be read as evidence of hardware feasibility.

All experiments are reproducible with the accompanying Python code and use no quantum software development kit. Exact circuit states are constructed directly, Pauli decompositions are computed explicitly, and finite-shot estimates are sampled from the exact Pauli means. The studies are designed to answer separate questions rather than to make a speedup claim.

Two six-DOF benchmarks are used and should not be confused. Benchmark~S is a fixed-base shear building with floor masses decreasing linearly from $220$ to $170$\,t and story stiffnesses from $220$ to $170$\,MN/m, giving $f_1=1.280$\,Hz; it is scaled by $s_t=U_G$ and padded with $\alpha=1.25$, and it backs the restart study of \cref{sec:results} and the certified-penalty run of Appendix~B. Benchmark~C is the $N=6$ member of the fixed-base chain family of Test~B, with $k=2.2\times10^8(1-0.35\xi)$\,N/m and element mass $2\times10^4(1-0.20\xi)$\,kg for $\xi\in[0,1]$ and lumped mass, giving $f_1=4.481$\,Hz; it is scaled by $s_t=\lambda_{\max}$ and padded with $\alpha=\max(1.10,U_G/s_t+10^{-3})$, and it backs Tests~A, D and E and the measurement floor. Both are one-dimensional and lumped-mass; the consistent-mass results are confined to Test~B.

\subsection{Test A: exact deflation threshold}

Benchmark~C is mass whitened and padded to three qubits. For targets $r=2,3,4$, a common penalty is swept across the exact threshold in \cref{eq:betathreshold}. Exact eigenvalues are used here only to verify the necessary-and-sufficient threshold; they are not counted as an operational preprocessing assumption. Thus $U_r=\varepsilon_r$ and $\delta_j=0$ in the threshold verification. In a deployable execution, Algorithm~\ref{alg:certvqe} instead uses either the Gershgorin policy \cref{eq:globalUr} or a recorded Ritz/Lanczos bound \cref{eq:ritzUr}. A second experiment rotates each prior projector by a controlled angle and compares the target-state rotation with \cref{eq:dkbound}.

\subsection{Test B: Hamiltonian and preprocessing scaling}

Fixed-base chains with $N\in\{4,8,16,32,64,128\}$ use smoothly varying stiffness and mass; this is the chain family of which Benchmark~C is the $N=6$ member, and it is the only test in which consistent mass appears. Lumped and consistent element mass matrices are compared. We record $U_G/\lambda_{\max}$, transformed density, nonzero Pauli terms, and greedy qubit-wise commuting groups. Cholesky and transformation structure are profiled through $N=512$. Wall-clock times are machine-specific and are not used as universal complexity claims.

\subsection{Test C: ansatz depth, optimizer restarts, and trainability}

Real-amplitude circuits use a CNOT ring. Fundamental-mode VQE is run for 8 and 16 DOFs at depths 1--3. The six-DOF first mode of Benchmark~S is repeated for 20 random initializations at depths 1 and 2 to distinguish systematic expressibility failure from local optimizer failure. Gradient variance is sampled at 2--12 qubits under uniform and near-identity initialization. The depth and gradient sweeps use chains from the Test~B family at the stated sizes, not either named six-DOF benchmark.

\subsection{Test D: finite-shot modal frequency}

For the exact first four modes of Benchmark~C, each Pauli term is measured independently with $10^2$, $10^3$, $10^4$, and $10^5$ shots. Monte Carlo estimates are converted to frequency using \cref{eq:freqrecover}. This isolates measurement propagation from ansatz and optimizer error.

\subsection{Test E: finite-shot damage localization}

Stiffness losses of 2\%, 5\%, 10\%, and 20\% are introduced at element 3 of Benchmark~C, counting from the base. The location is fixed rather than swept; across the six possible locations the exact normalized localization margin at 10\% loss ranges from 0.515 to 0.702, so element 3 at 0.654 is mid-range and not favourably selected. Baseline and damaged element energy fractions are estimated independently from Pauli measurements. Both estimates use the baseline element operators $\Hq_e$ of \cref{eq:He}; only the state and the normalization change between the healthy and damaged epochs. The element with maximum $D_e$ is selected. Each condition uses 200 Monte Carlo trials; random guessing has probability $1/6=16.7\%$.

\section{Results}
\label{sec:results}

\subsection{The theoretical penalty threshold is sharp}

\Cref{fig:penalty_threshold} shows that target-state overlap switches from zero to one at $\beta/\beta_{\min}=1$, up to the equality case where the ground state is degenerate. This numerical result verifies the theorem and replaces arbitrary tuning with a checkable spectral threshold.

\begin{figure}[t]
\centering
\includegraphics[width=\columnwidth]{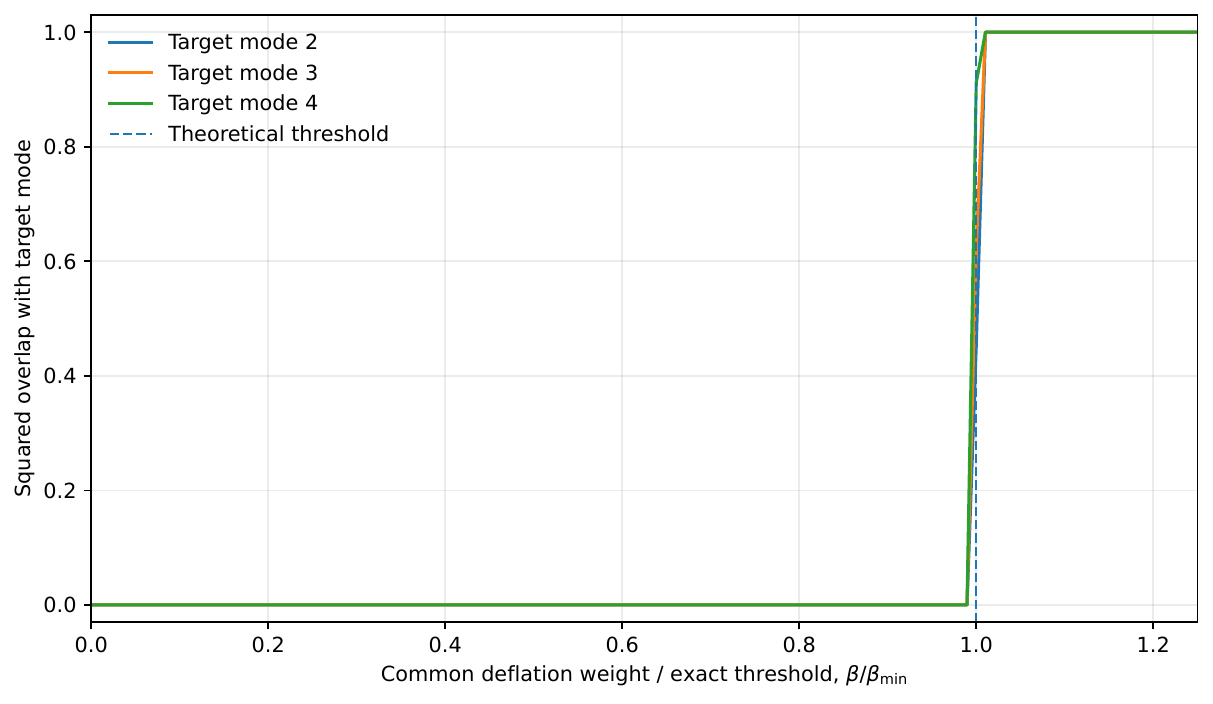}
\caption{Exact deflation sweep for target modes 2--4 on Benchmark~C, computed from the exact padded spectrum with no sampling. The target becomes the ground state only after the penalty exceeds the spectral threshold in \cref{eq:betathreshold}.}
\label{fig:penalty_threshold}
\end{figure}

Projector perturbations confirm the second half of the theory. The deflated separation is $g_r=0.05$, so \cref{eq:dkbound} applies while $\norm{\Delta_r}_2<g_r/2=0.025$. Within that range, projector rotation angles of 0.005 to 0.035 rad give perturbations $\norm{\Delta_r}_2=0.0034$ to $0.0239$ and target-state sine errors of 0.068 to 0.380, against Davis--Kahan bounds of 0.136 to 0.954. The bound holds at every point and is uniformly conservative by a factor of about two, correctly capturing the amplification of projector error by $\beta_j$. Beyond $\norm{\Delta_r}_2=g_r/2$ the guarantee lapses and the bound is vacuous: at a rotation of 0.06 rad the measured sine error is 0.508, but the target eigenvector is no longer certified to be isolated.

The perturbation is deliberately adversarial. The first projector is rotated toward the target mode itself, so the contamination lies exactly along the direction deflation is meant to protect; this is why a 0.02 rad projector error produces a 0.249 target rotation, a twelvefold amplification.

\subsection{Gershgorin scaling is not the dominant problem in the tested chains}

The empirical Gershgorin ratio is reported in \cref{tab:scaling}. Across both mass models and $N=4$--128, the ratio lies between 1.014 and 1.110. Thus the bound is mildly conservative for these one-dimensional chains, although no universal tightness is claimed.

The two-scale scheme remains useful because a different topology or material contrast can produce a looser bound, and because certified padding should not depend on the optimizer normalization.

\subsection{Mass whitening creates a measurable Pauli bottleneck}

\Cref{fig:density,fig:pauliscale} show the principal scaling result. At $N=128$, lumped mass yields transformed density 0.0233 and 576 nonzero Pauli terms; consistent mass yields density 0.3077 and 2950 terms. At $N=64$, greedy qubit-wise commuting groups increase from 64 for lumped mass to 221 for consistent mass.

\begin{figure}[t]
\centering
\includegraphics[width=\columnwidth]{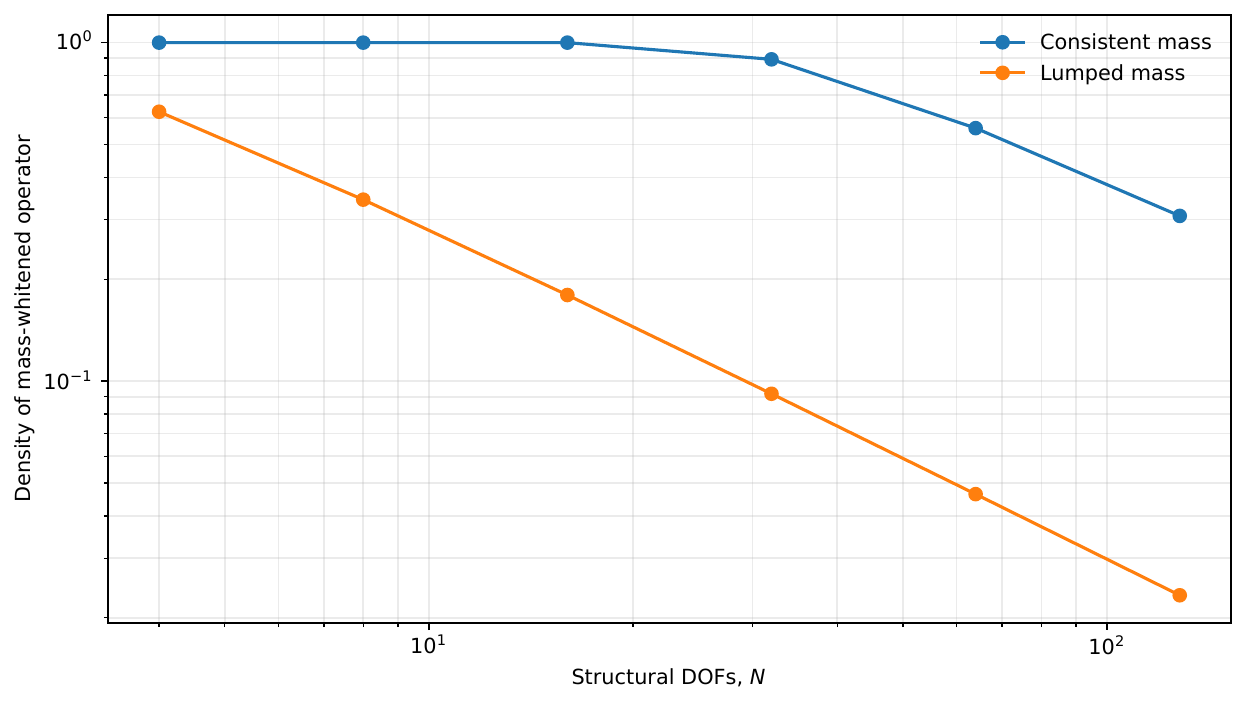}
\caption{Density of the mass-whitened operator for the fixed-base chain family of Test~B, evaluated exactly with no sampling, at the relative entry cut $\tau=10^{-12}$ of \cref{eq:density}. Consistent mass introduces substantially more long-range coupling than lumped mass in the tested chains.}
\label{fig:density}
\end{figure}

\begin{figure}[t]
\centering
\includegraphics[width=\columnwidth]{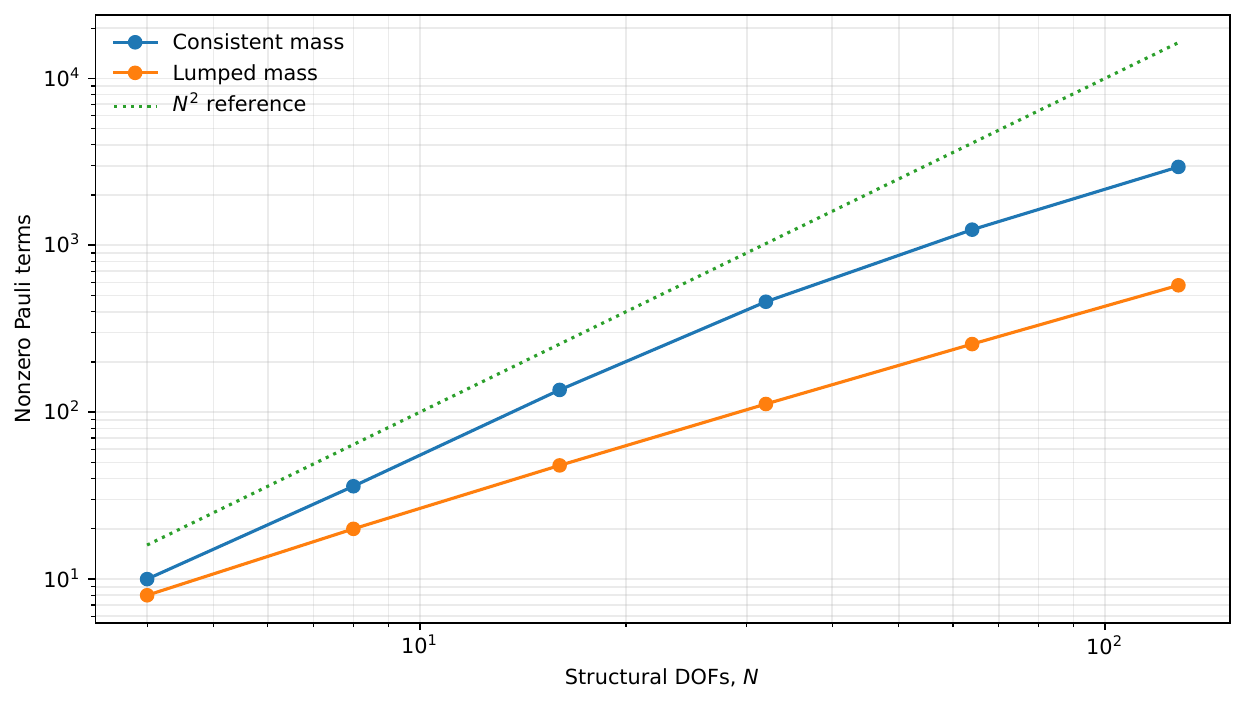}
\caption{Nonzero Pauli-term count versus structural DOFs for the chain family of Test~B, computed exactly at the absolute coefficient cut $10^{-10}$; the consistent-mass count is sensitive to that cut, as \cref{sec:resources} quantifies. The $N^2$ line is a reference, not a fitted asymptotic law.}
\label{fig:pauliscale}
\end{figure}

\begin{table}[t]
\centering
\caption{Selected Hamiltonian-scaling results. QWC denotes greedy qubit-wise commuting groups.}
\label{tab:scaling}
\begin{tabular}{lrrrrr}
\toprule
Mass & $N$ & $n_q$ & $U_G/\lambda_{\max}$ & Pauli terms & QWC groups\\
\midrule
Lumped & 16 & 4 & 1.053 & 48 & 16\\
Lumped & 64 & 6 & 1.022 & 256 & 64\\
Lumped & 128 & 7 & 1.014 & 576 & --\\
Consistent & 16 & 4 & 1.058 & 136 & 41\\
Consistent & 64 & 6 & 1.025 & 1241 & 221\\
Consistent & 128 & 7 & 1.017 & 2950 & --\\
\bottomrule
\end{tabular}
\end{table}

Preprocessing to $N=512$ confirms that the Cholesky factor of the chain remains banded, but $\A$ contains 1534 retained entries for lumped mass and 21554 for consistent mass. This is not evidence of universal dense scaling; it is direct evidence that whitening and mass formulation must be included in resource estimates.

\subsection{Circuit depth, not qubit count, controls exact-state recovery}

The 8- and 16-DOF fundamental-mode results are summarized in \cref{tab:vqe_depth} and \cref{fig:largervqe}. For 16 DOFs, depth 1 gives 81.75\% frequency error, depth 2 gives 5.16\%, and depth 3 gives $2.27\times10^{-9}$\%. The exact-state result shows that four qubits are enough to store the vector but not that a shallow circuit can prepare it.

\begin{table}[t]
\centering
\caption{State-vector VQE depth study for the fundamental mode. The failing rows are limited by ansatz expressibility and are stable under additional restarts; the near-zero entries are limited instead by optimizer tolerance and fall further as the restart budget is increased.}
\label{tab:vqe_depth}
\begin{tabular}{rrrrr}
\toprule
$N$ & Qubits & Depth & Parameters & Frequency error (\%)\\
\midrule
8 & 3 & 1 & 6 & 27.33\\
8 & 3 & 2 & 9 & $2.00\times10^{-12}$\\
8 & 3 & 3 & 12 & $8.59\times10^{-13}$\\
16 & 4 & 1 & 8 & 81.75\\
16 & 4 & 2 & 12 & 5.16\\
16 & 4 & 3 & 16 & $2.27\times10^{-9}$\\
\bottomrule
\end{tabular}
\end{table}

\begin{figure}[t]
\centering
\includegraphics[width=\columnwidth]{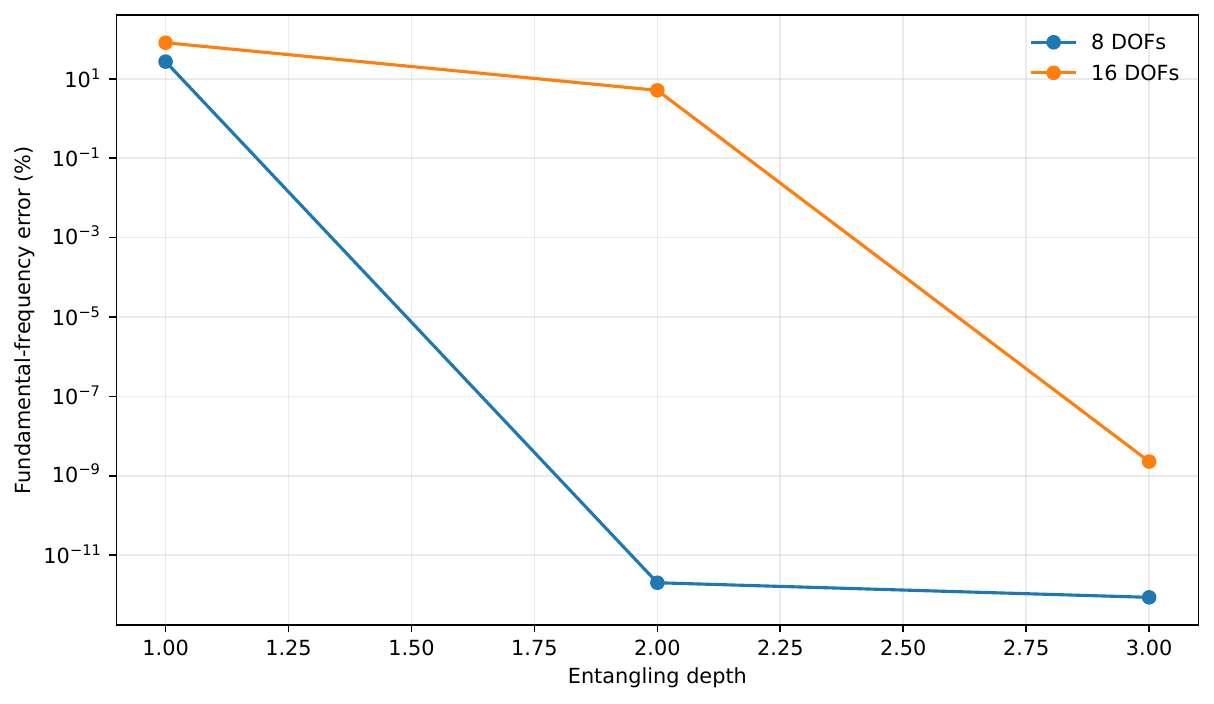}
\caption{Fundamental-frequency error versus circuit depth for 8- and 16-DOF members of the Test~B chain family. Exact state-vector optimization with no sampling noise, so the residual error is attributable to ansatz expressibility and the optimizer alone.}
\label{fig:largervqe}
\end{figure}

Twenty six-DOF restarts reinforce this interpretation. Depth 1 has a minimum first-mode error of 127.99\%, a median of 130.20\%, and a maximum of 130.20\%; depth 2 has a median error $4.1\times10^{-12}$\%. The shallow failure is therefore systematic ansatz error, not an unlucky optimizer seed.

\begin{remark}[Parameter-count lower bound on required depth]\label{rem:depth}
The depth dependence in \cref{tab:vqe_depth} follows from a parameter count. A real, normalized state in $D=2^{n_q}$ dimensions has $D-1$ free amplitudes, whereas the real-amplitude circuit of \cref{eq:ansatz} at depth $L$ has $n_q(L+1)$ parameters. Reaching an arbitrary target therefore requires
\begin{equation}
n_q(L+1)\ge 2^{n_q}-1,
\qquad\text{equivalently}\qquad
L\ge\frac{2^{n_q}-1}{n_q}-1.
\label{eq:depthbound}
\end{equation}
A parameter count bounds the dimension of the reachable manifold, so it establishes that a fixed-depth family cannot be surjective onto the state sphere; it does not by itself prove that one particular target is unreachable. The condition is nonetheless consistent with every row of \cref{tab:vqe_depth}: for $n_q=3$ it holds at $L=2$ ($9\ge7$) but not $L=1$ ($6<7$); for $n_q=4$ it holds at $L=3$ ($16\ge15$) but not $L=2$ ($12<15$). Since $2^{n_q}=N$, the required depth grows as $L=\Omega(N/\log_2 N)$, so a fixed-depth family cannot remain expressive as $N$ increases and the logarithmic qubit count does not translate into a logarithmic-depth circuit. The condition is necessary but not sufficient: the CNOT-ring entangler can further restrict the reachable manifold, so \cref{eq:depthbound} is a lower bound on the depth required for universal exact representation of arbitrary real states within this ansatz family. This is consistent with the systematic shallow-circuit failure reported above.
\end{remark}

\subsection{Gradient statistics raise, but do not settle, the barren-plateau question}

Two features of this diagnostic bound what it can establish, and both work against detecting a plateau rather than for it. The depth is held at $L=2$ while $n_q$ grows, so by \cref{rem:depth} the circuit is far from expressive over most of the range and comes nowhere near a unitary 2-design; the concentration mechanism that produces barren plateaus is therefore not something this experiment could exhibit even if it were operating. And the derivative is taken with respect to a single first-layer parameter, adjacent to $|0\rangle^{\otimes n_q}$, which is systematically the least plateau-prone choice available. The result below should be read as evidence that no exponential decay has set in for this particular shallow family and this particular parameter, not as a probe of the asymptotic regime.

\Cref{fig:gradients} reports gradient variance for a depth-2 circuit over 2--12 qubits, with bootstrap 95\% confidence intervals from the 80 samples per point. Under uniform initialization the variance declines from $2.62\times10^{-2}$ at 2 qubits to $9.82\times10^{-4}$ at 12 qubits. A power law in $n_q$ fits the uniform-initialization series better than an exponential: over the full range $R^2=0.96$ versus $0.83$, favouring the power law by $\Delta\mathrm{AIC}=15$; over the small-system-free range $n_q\ge4$ the advantage narrows to $R^2=0.90$ versus $0.85$ ($\Delta\mathrm{AIC}=4$, moderate). The decline does not stop: a linear fit to the tail ($n_q\ge7$) gives a shallow slope of $-0.12$ per qubit (95\% CI $[-0.23,-0.01]$), statistically nonzero but far below the per-qubit exponential rate that defines a barren plateau. Near-identity initialization produces smaller and nonmonotonic values, including $1.52\times10^{-5}$ at 12 qubits.

The decline is best read against the scale of the objective itself. Over the same range the variance of the energy across random parameters falls from $4.94\times10^{-2}$ to $2.32\times10^{-3}$, so the normalized quantity $\Var[\partial_a E]/\Var[E]$ is statistically flat: a linear fit of its logarithm gives a slope of $-0.036$ per qubit (95\% confidence interval $[-0.087,+0.016]$, $R^2=0.17$). Relative to the spread of the landscape it explores, the gradient does not shrink at all over $2\le n_q\le12$, and most of the raw decline in \cref{fig:gradients} reflects the concentration of the normalized spectrum rather than a loss of trainability. Thus at this fixed shallow depth the gradients shrink at an approximately polynomial, not exponential, rate.

That conclusion does not survive an increase in depth, and the fixed-depth diagnostic could not have revealed it. Repeating the sweep across $L\in\{2,4,8\}$, with the differentiated parameter drawn uniformly at random so the statistic is parameter-averaged rather than tied to $\theta_0$, the exponential decay rate steepens systematically: fitting $\log V_g$ against $n_q$ over $2\le n_q\le10$ gives $-0.42$ per qubit at $L=2$ (bootstrap 95\% confidence interval $[-0.49,-0.36]$ over 2000 resamples of the 80 draws at each $n_q$), $-0.62$ at $L=4$ ($[-0.67,-0.57]$) and $-0.73$ at $L=8$ ($[-0.77,-0.69]$), against the $-\ln 2\approx-0.69$ that characterizes $O(2^{-n_q})$ gradient concentration. The intervals separate the depths in both directions: at $L=2$ the decay is significantly shallower than that reference, whereas at $L=8$ the interval reaches it and the two are statistically indistinguishable. Restricting the fit to $n_q\ge7$ leaves the $L=8$ point estimate at $-0.78$ but widens the interval to $[-0.91,-0.64]$, so the restricted fit neither strengthens nor overturns the full-range one. The $L=8$ variance relative to $L=2$ falls monotonically for $n_q\ge5$, from $0.83$ to $0.12$, so this is a depth--width interaction and not an artefact of differentiating one of more parameters.

The implication is unfavourable. By \cref{rem:depth} \emph{universal} real-state expressibility requires $L=\Omega(N/\log_2N)$---at $n_q=10$ that is $L\ge101$, so $L=8$ lies far below the depth at which the family can cover a generic real state---yet a decay rate comparable to $2^{-n_q}$ is already measured there. Expressibility and trainability therefore pull against each other for this ansatz family on generic targets: the depths that could represent an arbitrary real state are far deeper than the depths at which gradient concentration has already set in. Two limits on this reading should be stated plainly. The parameter count bounds the dimension of the reachable manifold and so speaks to generic targets; it does not establish that the particular structural fundamental at $n_q=10$ is unreachable at $L=8$, only that the family is not surjective. And the sweep spans $n_q\le10$ for one ansatz family and one Hamiltonian family, so no asymptotic plateau is proved. What the result does remove is any reading in which the shallow-circuit success licenses optimism about deeper circuits.

\begin{figure}[t]
\centering
\includegraphics[width=\columnwidth]{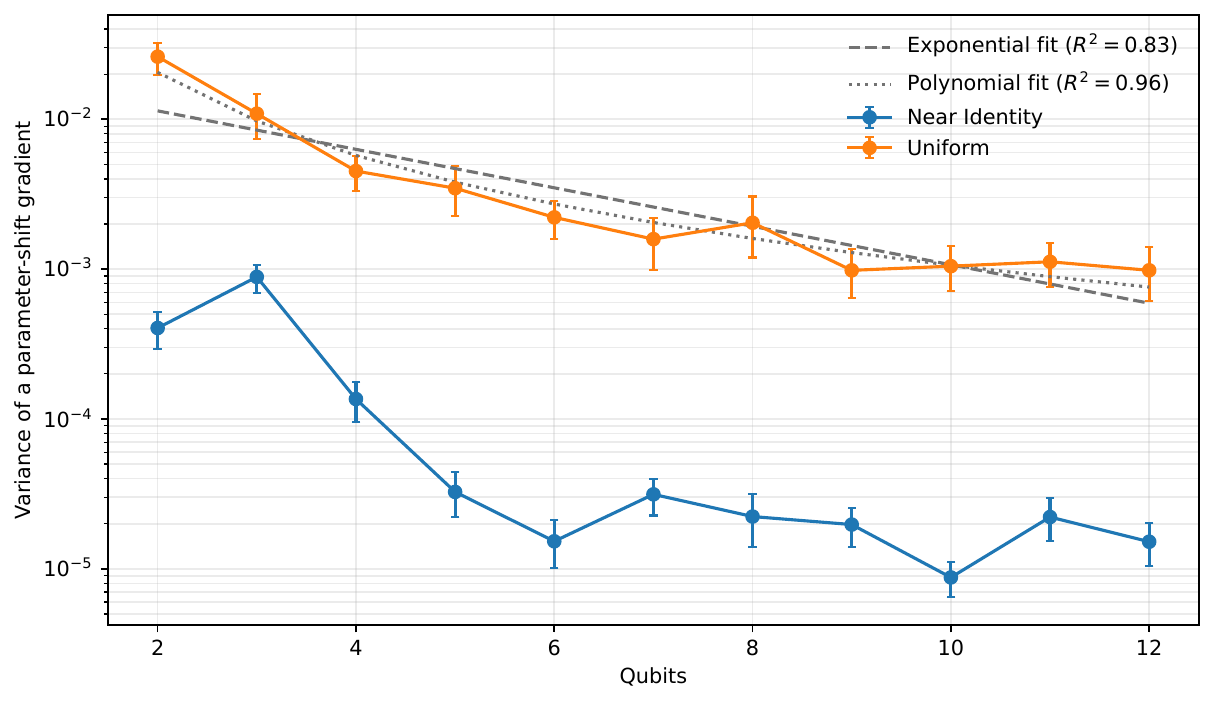}
\caption{Empirical gradient variance for depth-2 real-amplitude circuits on Test~B chains over 2--12 qubits, from 80 random parameter draws per qubit count, with bootstrap $95\%$ confidence intervals and polynomial and exponential fits to the uniform-initialization series. Exact state-vector gradients; the intervals reflect sampling of the parameter distribution, not shot noise. The variance is better fit by a power law than an exponential; the finite range is a diagnostic, not an asymptotic barren-plateau proof.}
\label{fig:gradients}
\end{figure}

\subsection{Finite-shot frequency estimation is most difficult for the first mode}

\Cref{tab:shotfreq} shows the independent-Pauli Monte Carlo results. At $10^4$ shots per Pauli term, frequency RMSE is 7.88\% for mode 1 but only 0.25\% for mode 4. Even at $10^5$ shots per term, mode-1 RMSE is 2.45\%. This follows from \cref{eq:freqvar}: the square-root transformation amplifies eigenvalue uncertainty at small $f_i$.

\begin{table*}[t]
\centering
\caption{Frequency RMSE (\%) under independent Pauli sampling and an optimistic QWC lower bound.}
\label{tab:shotfreq}
\begin{tabular}{lrrrr}
\toprule
Measurement budget & Mode 1 & Mode 2 & Mode 3 & Mode 4\\
\midrule
$10^2$ shots/term, independent & 64.20 & 10.72 & 4.49 & 2.57\\
$10^3$ shots/term, independent & 25.73 & 3.32 & 1.37 & 0.79\\
$10^4$ shots/term, independent & 7.88 & 1.00 & 0.47 & 0.25\\
$10^5$ shots/term, independent & 2.45 & 0.36 & 0.13 & 0.08\\
\midrule
QWC-optimal floor at matched total budget$^{\dagger}$ & 0.83 & 0.14 & 0.054 & 0.032\\
\bottomrule
\end{tabular}

\vspace{0.35em}
\parbox{0.97\textwidth}{\footnotesize $^{\dagger}$Benchmark~C's padded three-qubit Hamiltonian has 19 nonidentity Pauli terms and eight greedy qubit-wise-commuting groups; that count is insensitive to the coefficient cut, being 19 at every cut from $10^{-8}$ to $10^{-14}$. The frequency recovery \eqref{eq:freqrecover} clips a negative energy estimate to zero. That clip activates in 23.8\% of mode-1 trials at $10^2$ shots per term and 0.8\% at $10^3$, and in no trial of any mode at $10^4$ or $10^5$ shots: the two highest-budget rows are therefore unclipped plug-in estimates, while the $10^2$ mode-1 entry is biased low by construction and indicates the trend only. The independent rows are 400-trial Monte Carlo estimates, carrying a relative standard error of about 3.5\% (so the mode-1 entry is good to roughly $\pm0.09$); the floor is analytic, and the improvement factors quoted in the text compare it against the \emph{analytic} independent RMSE (2.285\% for mode~1), not against the Monte Carlo row above, so dividing the tabulated rows directly gives a slightly different ratio. The floor uses the same $1.9\times10^6$ total shots per mode as the $10^5$-shots-per-term baseline, exact within-group covariances, continuous optimal allocation, perfect simultaneous measurement, and no circuit, readout, or allocation overhead. It is therefore an optimistic lower bound, not an expected hardware result.}
\end{table*}

The independent experiment is a baseline, not an optimized measurement strategy. At matched total state preparations, exact-covariance QWC grouping lowers the theoretical RMSE floor by factors of approximately 2.5--2.8, but the first-mode floor remains 0.83\%. Thus grouping is valuable yet does not remove the low-frequency sensitivity. Classical shadows \cite{Huang2020}, hardware-efficient entangled measurements, and error mitigation \cite{Endo2021} should be compared in future work.

\subsection{Finite-shot damage localization is estimable, but only given the damaged model}

Exact-state agreement of the modal energy fractions with the classical solution is a consistency check only. The finite-shot localization study is more informative. \Cref{fig:shotdamage} shows that at $10^5$ shots per Pauli term a 10\% stiffness loss is localized correctly in 84.5\% of 200 trials (Wilson 95\% interval approximately 78.8--88.9\%), against a 16.7\% random baseline, and a 20\% loss in every trial. The median localization margin is positive at this budget ($+0.28$ for 10\% loss), meaning the damaged element typically outscores every intact competitor.

The energy-fraction ratio floors its denominator at $10^{-10}$; that floor activates in 4.9--6.8\% of samples at $10^2$ shots per term, below 0.5\% at $10^3$, and in none at $10^4$ or $10^5$, so the two budgets that carry the conclusions are unaffected. The budget dependence is steep. The same 10\% loss is localized in only 40.5\% of trials at $10^4$ shots per term (33.9--47.4\%) and remains only marginally above guessing at $10^3$ (23.0\%, Wilson 17.7--29.3\%; one-sided binomial $p=0.013$ against $1/6$ on the same $n=200$ trials). Small damage remains hard at every budget tested: a 2\% loss reaches only 32.0\% at $10^5$ shots per term (Wilson 25.9--38.8\%), with a still-negative median margin. Localization therefore succeeds in this benchmark, but only for moderate damage and only at measurement budgets that \cref{sec:discussion} quantifies as the dominant end-to-end cost.

\begin{figure}[t]
\centering
\includegraphics[width=\columnwidth]{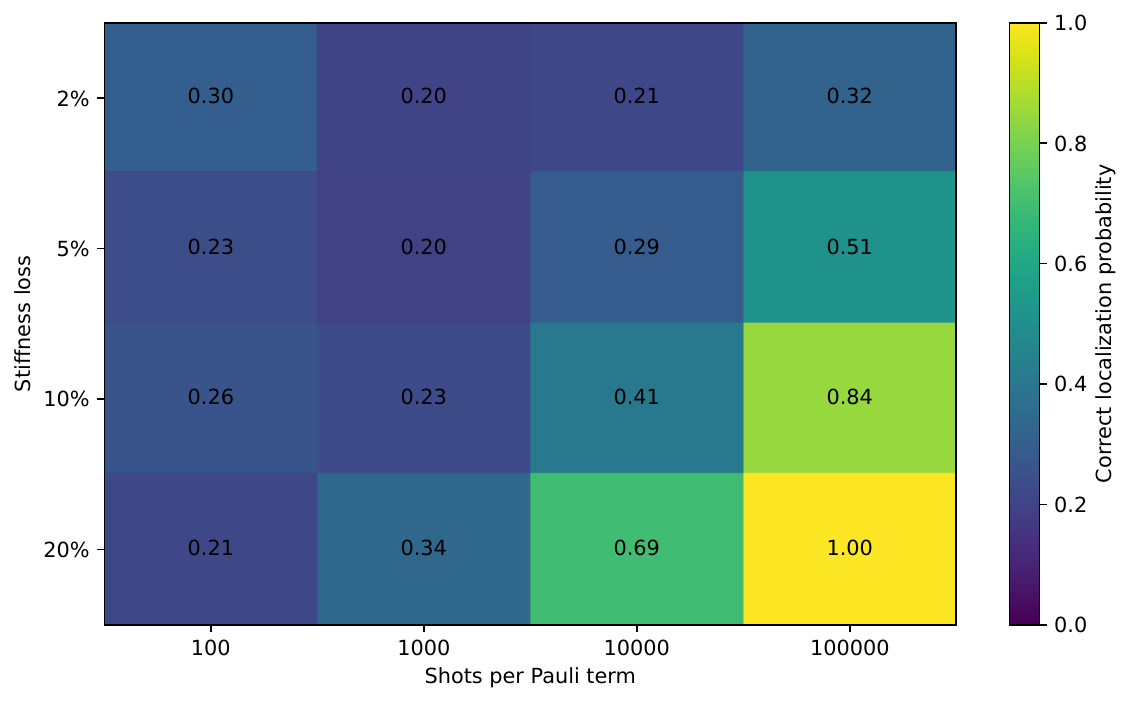}
\caption{Correct damage-localization probability from independently sampled element-energy observables, for Benchmark~C at 200 Monte Carlo trials per condition. The modal states are the \emph{exact} healthy and damaged eigenvectors, not variationally obtained ones, so these curves isolate measurement noise and are an upper bound on what an end-to-end pipeline would achieve. The random-guess baseline is $1/6=16.7\%$.}
\label{fig:shotdamage}
\end{figure}

\begin{remark}[Sensitivity to the element-operator convention]\label{rem:convention}
The choice in \cref{eq:mse} to build the numerator from the baseline $\K_e^0$ is consequential, not cosmetic. If the operator is instead built from the damaged element stiffness $\K_e(\mathbf d)$---an oracle that presumes the quantity being inferred---the explicit $(1-d_e)$ factor opposes the denominator and mode-shape changes and cancels most of the signal: the exact localization margin at 10\% loss falls by a factor of 11.2, measured accuracy drops from 84.5\% to 25.0\% (Wilson 19.5--31.4\%), and the design shot requirement rises by two orders of magnitude. Estimating a difference of operator expectations is therefore sensitive to which operator difference is taken, a point that applies to quantum change-detection observables generally.
\end{remark}

A cross-check against the frequency results of \cref{sec:results} sharpens the picture. The same 10\% loss shifts the modal frequencies by $-1.20$\%, $-0.16$\%, $-1.75$\% and $-0.24$\% for modes 1--4, while the corresponding finite-shot RMSEs at $10^5$ shots per Pauli term are 2.45\%, 0.36\%, 0.13\% and 0.08\%. The damage signal is thus below the measurement noise for the two lowest modes and an order of magnitude above it for the third. Frequency-based detection of this damage is possible at this budget, but only through the higher modes; the fundamental, on which most vibration-based monitoring practice rests, is the least usable of the four.

The study tests the damage formulation end to end and identifies measurement variance, not the formulation, as the limiting factor. Direct independent-Pauli estimation remains inefficient for small stiffness changes; grouped or shadow-based measurement of the change observable is the obvious next step.

\section{Generalization to two and three dimensions}
\label{sec:dimensions}

Every experiment above uses a one-dimensional chain, the most favorable case for Cholesky fill. This section tests whether the mapping and deflation theory carry over unchanged to higher spatial dimension, and whether the densification finding of \cref{sec:mapping,sec:resources} is a property of the 1D chain specifically or worsens as dimension increases.

\subsection{What remains dimension-agnostic}

Nothing in the theory assumes a specific spatial dimension. The Cholesky mapping and mass--Euclidean isometry (Proposition~1) hold for any $\M\succ0$, $\K=\K^T$. Spectral-safe padding (Proposition~2) and the scalar-scaling law (Proposition~3) depend only on the spectrum of $\A$. The exact deflation theorem and the Davis--Kahan perturbation bound (Theorem~1, \cref{eq:dkbound}) hold for any Hermitian operator with a spectral gap, independent of what physical system it was built from. The parameter-count depth bound (\cref{rem:depth}, \cref{eq:depthbound}) depends only on $n_q$, not on how the DOFs are arranged in space. The resource vector \cref{eq:resourcevector} and the finite-shot propagation formulas (\cref{eq:paulivar,eq:energyvar,eq:freqvar}) are likewise generic to any Pauli-decomposed Hamiltonian. Only the \emph{values} taken by $N_P$, the transformed density, and the Gershgorin ratio are specific to a given mesh, and those are what this section measures.

\subsection{Two- and three-dimensional scaling}
\label{sec:dim-scaling}

\begin{table*}[t]
\centering
\caption{Two- and three-dimensional scalar-mesh scaling, generalizing \cref{tab:scaling}. Grid sizes are chosen as exact powers of two, as in the 1D study. Densities of exactly $1.000$ are rounded rather than literal; the exact deficits are given in \cref{sec:dim-scaling}, where the 3D power-law exponent quoted is also noted to rest on three sizes only.}
\label{tab:dimscaling}
\begin{tabular}{llrrrr}
\toprule
Dimension & Mass & $N$ & $n_q$ & Density of $\A$ & Pauli terms\\
\midrule
2D & Lumped & 16 & 4 & 0.250 & 48\\
2D & Lumped & 64 & 6 & 0.070 & 256\\
2D & Lumped & 256 & 8 & 0.019 & 1280\\
2D & Lumped & 1024 & 10 & 0.005 & 6144\\
2D & Consistent & 16 & 4 & 1.000 & 136\\
2D & Consistent & 64 & 6 & 1.000 & 2079\\
2D & Consistent & 256 & 8 & 0.947 & 25082\\
2D & Consistent & 1024 & 10 & 0.433 & 147620\\
3D & Lumped & 64 & 6 & 0.086 & 256\\
3D & Lumped & 128 & 7 & 0.045 & 576\\
3D & Lumped & 1024 & 10 & 0.006 & 6144\\
3D & Consistent & 64 & 6 & 1.000 & 2080\\
3D & Consistent & 128 & 7 & 1.000 & 8251\\
3D & Consistent & 1024 & 10 & 0.908 & 312325\\
\bottomrule
\end{tabular}
\end{table*}

The mesh is a scalar grid of point masses connected by nearest-neighbor springs, the direct $d$-dimensional generalization of the 1D chain: the same two-node element pattern is applied along $d$ axes instead of one, with one DOF per node so that the comparison isolates spatial dimension rather than DOFs per node. Free-DOF counts are fixed at exact powers of two by fixing an entire boundary layer along one axis, generalizing the chain's fixed base. Stiffness and mass vary smoothly and asymmetrically across the grid axes, exactly as the 1D study varies properties along the chain, so that no accidental grid symmetry inflates or deflates the sparsity count. \Cref{tab:dimscaling} reports density and nonzero Pauli terms for lumped and consistent mass at matched qubit counts.

Lumped mass remains markedly sparser than consistent mass at every tested size in every dimension, as in the 1D case. Densities of exactly 1.000 in \cref{tab:dimscaling} are rounded: the 2D $N=64$ and 3D $N=128$ operators carry 2079 and 8251 nonzero Pauli terms against the real-symmetric maxima 2080 and 8256, the small deficits being accidental cancellations rather than structural sparsity. For consistent mass, the informative comparison is at matched $N$: at $N=1024$, 2D density has fallen to 0.433 while 3D remains at 0.908---the same DOF count, a visibly different regime. This is consistent with nested-dissection theory: 2D separators have size $O(\sqrt N)$, giving fill $O(N\log N)$ \cite{GeorgeNg1988}; the corresponding 3D separator size $O(N^{2/3})$ and fill $O(N^{4/3})$ follow from the same separator argument but lie outside the two-dimensional scope of that reference, and are quoted here as the standard extrapolation rather than as a cited result. On either reading, so the onset of sparsity-recovering behavior is expected to shift to larger $N$ as dimension increases, and 3D should lag furthest behind. Fitting a single power law for the nonzero Pauli count over the full tested range gives an increasing but modest exponent (1D $\approx\!1.66$, 2D $\approx\!1.69$, 3D $\approx\!1.79$; the 3D fit rests on only three sizes). These are full-range fits: restricting the 1D fit to $N\ge16$ gives 1.48 instead, which is the figure quoted in \cref{sec:complexity}, and the 1D local exponent continues to fall with $N$. Because density has not yet stabilized into a clean power-law regime in any dimension over the tested sizes, these exponents are reported as illustrative fits, not asymptotic laws, exactly as \cref{sec:results} already cautions for the 1D case. The Gershgorin ratio remains mild in every case tested (1.04--1.14), consistent with the 1D finding of \cref{tab:scaling} that Gershgorin looseness is not the dominant issue.

Declining density is a fractional effect: the absolute nonzero Pauli-term count continues to grow in every dimension tested (\cref{tab:dimscaling}), so this does not indicate that the measurement burden diminishes with $N$. Density and $N_P$ are structural properties of the operator, not complexities in themselves; they separately drive the space row and the measurement term of the time row in \cref{tab:complexity}. The practical relevance of eventual sparsity recovery is also limited by scale. For an idealized cubic mesh, the graph diameter (maximum node-to-node hop distance) is $d(N^{1/d}-1)$ in $d$ dimensions; matching the diameter of the $N=128$ 1D chain already tested (127 hops) would require an idealized square mesh of $N\approx4{,}160$ (12 qubits) or an idealized cubic mesh of $N\approx81{,}000$ (16.3 qubits). This is an illustrative order-of-magnitude estimate for a regular cubic mesh, not a claim about arbitrary geometries, but it places any such recovery for 3D structures far beyond near-term hardware and, likely, beyond the regime in which classical simulation of the whitened operator itself remains convenient.

These results hold under the locality-preserving ordering used throughout, and do not extend to irregular topology, block DOFs, realistic 3D frames, or alternative orderings; they show that the 1D densification finding is not an artifact of one spatial dimension, and that it does not improve---if anything it worsens and appears later---as dimension increases.

\subsection{Fill-reducing orderings trade one sparsity for the other}
\label{sec:ordering}

Every count reported above uses the locality-preserving row-major indexing, and a referee may
reasonably ask whether the dimensional deterioration is an artefact of that choice rather than
of dimension. It is not---but the ordering matters more than the dimension does, and in the
direction opposite to classical intuition. \Cref{tab:ordering} compares three symmetric
permutations of the structural graph at fixed $N$, fixed dimension, and fixed physics, since a
symmetric permutation changes no eigenvalue: the natural indexing, reverse Cuthill--McKee
(RCM), a standard bandwidth- and profile-reducing sparse-matrix ordering, and a fixed-seed random
scramble included as a worst-case bracket. RCM is not a fill-minimizing ordering in the sense of
AMD or nested dissection; it is used here because it is the common default, and it does reduce
the consistent-mass factor fill measurably below.

\begin{table}[t]
\centering
\caption{Effect of reordering on Cholesky fill and on measurement cost, at fixed $N$, fixed
dimension and fixed spectrum. $\mathrm{nnz}(\mathbf L)$ is the fill of the Cholesky factor of
$\M$, $\rho_\A$ the density \eqref{eq:density} of the whitened operator, and the last column
gives $N_P$ relative to the natural ordering. Lumped-mass $\mathrm{nnz}(\mathbf L)$ and
$\rho_\A$ are permutation-invariant, so those rows isolate the effect on Pauli locality alone.}
\label{tab:ordering}
\small
\setlength{\tabcolsep}{4pt}
\begin{tabular}{llrrrr}
\toprule
Mass & Ordering & $\mathrm{nnz}(\mathbf L)$ & $\rho_\A$ & $N_P$ & ratio\\
\midrule
\multicolumn{6}{l}{\emph{2D, $N=256$ ($n_q=8$)}}\\
Lumped     & natural & 256   & 0.019 & 1280   & 1.00\\
Lumped     & RCM     & 256   & 0.019 & 7168   & 5.60\\
Lumped     & random  & 256   & 0.019 & 28288  & 22.10\\
Consistent & natural & 3967  & 0.947 & 25082  & 1.00\\
Consistent & RCM     & 2845  & 0.948 & 32859  & 1.31\\
Consistent & random  & 5094  & 0.540 & 32896  & 1.31\\
\midrule
\multicolumn{6}{l}{\emph{3D, $N=512$ ($n_q=9$)}}\\
Lumped     & natural & 512   & 0.012 & 2816   & 1.00\\
Lumped     & RCM     & 512   & 0.012 & 39808  & 14.14\\
Lumped     & random  & 512   & 0.012 & 123136 & 43.73\\
Consistent & natural & 29619 & 1.000 & 113522 & 1.00\\
Consistent & RCM     & 19566 & 1.000 & 131326 & 1.16\\
Consistent & random  & 43942 & 0.808 & 131328 & 1.16\\
\bottomrule
\end{tabular}
\end{table}

Three readings follow. First, the natural ordering is not merely convenient, it is optimal in
the Pauli metric among those tested: the lumped-mass counts $1280$ and $2816$ attain the bound
of \cref{prop:lumped} exactly, $(n_q+2)2^{\,n_q-1}$, and no ordering tested improves on the
natural count in any of the four groups. Second, RCM does precisely what it is designed to do,
reducing the consistent-mass Cholesky fill by $28\%$ in 2D and $34\%$ in 3D, and is paid for in
measurement: $N_P$ rises by $31\%$ and $16\%$ respectively.

Third, and most sharply, the lumped-mass rows isolate the mechanism. There $\M$ is diagonal, so
$\mathbf L$ is diagonal and $\mathrm{nnz}(\mathbf L)=N$ under \emph{every} ordering, and
$\rho_\A$ is likewise permutation-invariant: a fill-reducing permutation has nothing whatever
to reduce. Yet $N_P$ rises by a factor of $5.6$ in 2D and $14.1$ in 3D under RCM, and by
$22.1$ and $43.7$ under a random scramble. Matrix sparsity and Pauli locality are therefore not
two views of one property. The first is invariant under symmetric permutation and the second is
not, and only the second controls the measurement cost.

The random rows bracket the range. The counts $32896$ and $131328$ are exactly the
real-symmetric maxima $2^{2n_q-1}+2^{n_q-1}$ of \cref{sec:resources}, so a scrambled ordering
does not merely worsen the count---it saturates the worst case. The full span from the
$O(N\log N)$ floor of \cref{eq:lumpedcount} to the $\Theta(N^2)$ ceiling is thus reachable by
permutation alone, at fixed $N$, fixed dimension and identical physics, which is a wider range
than anything attributable to spatial dimension in \cref{tab:dimscaling}.

The practical consequence inverts the classical preprocessing intuition. In a sparse classical
eigensolver one applies a fill-reducing ordering before factorizing; here that same step,
applied before the Pauli decomposition, increases the measurement cost---and in the lumped case
it does so while delivering no fill benefit at all. The $O(N\log N)$ result of
\cref{prop:lumped} is therefore contingent on \emph{not} applying standard classical
preprocessing, which should be stated as a condition on the result rather than left implicit.
This also supplies the evidence for the ordering caveat attached to \cref{eq:lumpedcount} in
\cref{sec:resources}.

\subsection{Validating the degeneracy diagnostics}
\label{sec:dim-degeneracy}

A 1D chain with smoothly varying properties has no repeated eigenvalues and so can never exercise the subspace-based diagnostic of \cref{eq:subspace}. A uniform two-dimensional grid can: it carries exactly repeated eigenvalues, although, as \cref{sec:neardegenerate} establishes, that repetition is algebraic rather than forced by the grid's diagonal reflection symmetry. Using a uniform (not the heterogeneous) grid from \cref{sec:dim-scaling}, an adjacent-eigenvalue pair was found with relative gap $2.2\times10^{-16}$, i.e.\ degenerate to floating-point precision.

For a degenerate pair, any rotation of the two eigenvectors within their shared subspace is an equally valid choice, so two independent solves (or a solver and a slightly perturbed rerun) may return differently oriented, but physically identical, bases. Let $\ph_a(\theta)=\cos\theta\,\ph_a+\sin\theta\,\ph_b$ be the pair rotated by angle $\theta$ within the degenerate subspace. Comparing $\ph_a$ to $\ph_a(\theta)$ via \cref{eq:mac} gives, since both are $\M$-orthonormal,
\begin{equation}
\MAC_{aa}^{(M)}(\theta)=\cos^2\theta,
\label{eq:macrotation}
\end{equation}
exactly---confirmed numerically to machine precision. This degrades from $1$ to $0.5$ at $\theta=45^\circ$ and to essentially $0$ at $\theta=90^\circ$, which by conventional MAC thresholds would be flagged as a severe mismatch, even though the physical subspace never changed. The subspace-angle diagnostic \cref{eq:subspace} is stated in the whitened representation, where the isometry of Proposition~1 makes a plain transpose product equivalent to an $\M$-weighted one; written directly in the physical representation, it evaluates $\varepsilon_{\mathrm{sub}}$ between $\mathrm{span}(\ph_a,\ph_b)$ and its rotated copy via $(\mathbf Y_C)^T\M\mathbf Y_Q$. Since the rotation is an isometry of the shared subspace, this overlap matrix is exactly a $2\times2$ rotation matrix, whose singular values are exactly one, giving $\varepsilon_{\mathrm{sub}}=0$ for every $\theta$. The measured values stay at the eigensolver's numerical floor ($\sim\!3\times10^{-8}$) across the full sweep. Since $\varepsilon_{\mathrm{sub}}=0$ follows analytically, this confirms the implementation rather than discovering the behaviour. The informative case---a \emph{near}-degenerate pair under a physical perturbation, where $\varepsilon_{\mathrm{sub}}$ is small but nonzero---is treated in \cref{sec:neardegenerate}.

This is a direct demonstration of why \cref{sec:observables}'s subspace comparison exists: a per-mode diagnostic can manufacture a false structural-health alarm from an arbitrary, physically meaningless basis choice whenever a structure has near-symmetric geometry, and the subspace-angle diagnostic is immune to exactly that failure mode by construction.

\subsection{Near-degenerate pairs and the gap at which the certificate lapses}
\label{sec:neardegenerate}

The exactly degenerate case above is analytically forced. The decision-relevant case is the
near-degenerate one, and it is reached by breaking the reflection symmetry physically: the
axis-0 springs are stiffened by a fraction $\delta$, which splits the pair with a relative gap
proportional to $\delta$. A $3\times3$ uniform grid pinned at one corner gives $N=8=2^3$ free
DOFs, so no spectral padding is required and $G_r$ in \cref{eq:physgap} reduces to the physical
gap $\varepsilon_{r+1}-\varepsilon_r$. The grid carries an exact repeated eigenvalue that the anisotropy
splits, but its origin is not a point-group degeneracy. The only operations of the square's
symmetry group that fix the pinned corner and preserve both $\K$ and $\M$ are the identity and
the diagonal reflection, and that $\mathbb Z_2$ admits only one-dimensional irreducible
representations, so it cannot force a twofold degeneracy. Consistently, the reflection acts as
$\mathrm{diag}(-1,+1)$ on the pair, placing the two modes in \emph{different} one-dimensional
representations; and the repeated eigenvalue sits at the same $\lambda=6.6\times10^{4}$ for each
$n_{\mathrm{side}}\in\{3,4,5,6,7\}$ tested here, with multiplicity $n_{\mathrm{side}}-1$, which no point-group argument
produces. The degeneracy is therefore algebraic rather than symmetry-protected. What the sweep
requires is only that it is exact and that a directional stiffness anisotropy splits it smoothly.

At every value of $\delta$ the \emph{same} small model discrepancy is applied---a $1\%$
stiffness loss on one spring---so the perturbation size $\norm{\Delta_r}_2$ is held essentially
constant ($2.10$--$2.15\times10^{-3}$ across the sweep) while the gap varies by a factor of
$50$. This isolates the effect of the gap alone. \Cref{tab:neardeg} reports the outcome.

\begin{table}[t]
\centering
\caption{A fixed $1\%$ stiffness discrepancy applied to a near-degenerate pair whose gap is
controlled by an axis asymmetry $\delta$. The perturbation is held constant
($\norm{\Delta_r}_2=2.10$--$2.15\times10^{-3}$ in the normalized representation); only the gap
changes. $\mathrm{DK}_{\mathrm{ind}}=2\norm{\Delta_r}_2/G_r$ and
$\mathrm{DK}_{\mathrm{sub}}=2\norm{\Delta_r}_2/g_{\mathrm{pair}}$ are the right-hand sides of
\cref{eq:dkbound} for the individual mode and for the pair treated as a block; a value
$\ge1$ is vacuous, since a sine never exceeds one. Probe: the interior axis-0 spring joining
nodes $(1,1)$ and $(2,1)$. $^{\ast}$Relative gap and $G_r$ are tabulated in units of
$10^{-3}$.}
\label{tab:neardeg}
\small
\setlength{\tabcolsep}{4pt}
\begin{tabular}{rrrrrrr}
\toprule
$\delta$ (\%) & gap$^{\ast}$ & $G_r^{\ast}$ & $\MAC_{rr}$ & $\varepsilon_{\mathrm{sub}}$
& $\mathrm{DK}_{\mathrm{ind}}$ & $\mathrm{DK}_{\mathrm{sub}}$\\
\midrule
0.1 & 0.82  & 0.22  & 0.026 & 0.0032 & 18.70 & 0.034\\
0.2 & 1.63  & 0.45  & 0.147 & 0.0032 & 9.36  & 0.034\\
0.5 & 4.06  & 1.12  & 0.990 & 0.0032 & 3.75  & 0.034\\
1.0 & 8.09  & 2.24  & 0.999 & 0.0032 & 1.89  & 0.034\\
2.0 & 16.04 & 4.45  & 1.000 & 0.0032 & 0.95  & 0.035\\
5.0 & 39.05 & 10.97 & 1.000 & 0.0033 & 0.39  & 0.036\\
\bottomrule
\end{tabular}
\end{table}

Three things follow. First, the per-mode diagnostic fails \emph{continuously} as the gap
closes, not abruptly: $\MAC_{rr}$ falls from $1.000$ to $0.026$ as the relative gap shrinks
from $3.9\times10^{-2}$ to $8.2\times10^{-4}$, crossing the conventional $0.9$ acceptance
threshold between relative gaps of $1.6\times10^{-3}$ and $4.1\times10^{-3}$. At the tightest
gap the two modes are nearly orthogonal, from a $1\%$ change in one spring. Second,
$\varepsilon_{\mathrm{sub}}$ is flat to three significant figures at $0.0032$ over the entire
sweep: the subspace diagnostic responds to the size of the discrepancy and not at all to the
proximity of the pair, which is the property \cref{eq:subspace} was introduced for and which
the exactly degenerate test could not show.

Third, and connecting to \cref{sec:deflation}, the deflation certificate lapses at a gap that
can be stated. Because $g_r\le G_r$ by \cref{eq:deflgap}, no choice of penalties can rescue the
target-mode bound once $G_r$ closes; even the certificate-optimal $\beta_j^\star$ of
\cref{eq:betastar} only attains $g_r=G_r$. The individual bound is therefore vacuous for
$\delta\le1\%$ and non-vacuous for $\delta\ge2\%$ in \cref{tab:neardeg}. Note that Weyl's
precondition $\norm{\Delta_r}_2<g_r/2$ and non-vacuity of \cref{eq:dkbound} are the same
inequality, so the bound does not degrade gracefully---it simply switches off. The pair treated
as a block behaves quite differently: its separation from the rest of the spectrum,
$g_{\mathrm{pair}}=\min(\varepsilon_r-\varepsilon_{r-1},\varepsilon_{r+2}-\varepsilon_{r+1})$,
stays near $0.12$ throughout, and $\mathrm{DK}_{\mathrm{sub}}$ never exceeds $0.054$ anywhere in
the sweep. Where the individual certificate is unavailable, the block certificate remains valid
by a margin of more than an order of magnitude.

These conclusions do not depend on where the discrepancy is placed. Sweeping the probe over all
$12$ springs, the individual certificate is vacuous at \emph{every} probe location for
$\delta=0.1\%$ ($\mathrm{DK}_{\mathrm{ind}}\in[8.8,27.7]$) and non-vacuous at every location for
$\delta=5\%$ ($[0.18,0.58]$), while $\varepsilon_{\mathrm{sub}}$ stays within
$[0.0021,0.0043]$ across all $72$ combinations and $\mathrm{DK}_{\mathrm{sub}}$ within
$[0.034,0.054]$. The measured rotations respect both bounds in all $72$ cases. The practical
reading for \cref{rem:degenerate} is that a near-degenerate pair should be deflated and reported
as a block, and that the individual certificate be trusted only while $2\norm{\Delta_r}_2<G_r$, which is
exactly the condition under which \cref{eq:dkbound} is non-vacuous. In this sweep that boundary
happens to fall where the relative gap is within about an order of magnitude of the model
discrepancy, but that correspondence is an empirical feature of these runs rather than a
general rule.

\section{Discussion}
\label{sec:discussion}

\subsection{What is now theoretically guaranteed}

Three properties have exact statements. First, Cholesky whitening preserves eigenvalues and converts mass orthogonality to Euclidean orthogonality. Second, safe padding places artificial states above a certified physical bound for any positive optimization scale. Third, exact deflation identifies mode $r$ if and only if every lower-mode penalty exceeds its spectral difference from the target. With approximate prior states, \cref{eq:deflperturb,eq:dkbound} quantify the resulting loss of accuracy.

No theorem is claimed for global convergence of the parameterized circuit optimizer. This distinction is essential: a correct deflated operator can still be inaccessible to a shallow ansatz or missed by a nonconvex optimizer.

\subsection{Why heuristic composite penalties are unsafe}

A weighted cost $E+\gamma\sigma_H^2+\zeta\ell_p$ changes the optimization landscape without a universal conversion between energy, residual, and leakage. An adaptive multiplicative $\beta$ rule is uncontrolled in both directions: it can stall just above the spectral threshold, where \cref{eq:dkbound} is vacuous, or grow without reference to $G_r$, where the certificate degrades linearly. Temporal overlap and frequency penalties can suppress genuine damage-induced changes. A scalar exponential integrity score can conceal a severe failure in one component behind good performance in another. Finally, a mixed-dimension ``information efficiency'' is uninterpretable without a reference task and cost model.

The proposed framework therefore uses theorem-based penalties, constrained diagnostics, independent epoch solutions, non-aggregated validation vectors, and raw resource counts.

\subsection{The actual scaling bottleneck}

For the tested chain models, Gershgorin looseness is modest. The more serious issue is that consistent-mass whitening changes operator structure, increasing nonzero Pauli terms and measurement groups; \cref{sec:dimensions} shows this bottleneck does not ease, and its onset shifts to larger $N$, as spatial dimension increases. The density of $\A$ is not an independent design variable: for fixed $\K$ and DOF ordering it is a deterministic consequence of the mass model through the Cholesky factor. The design tension is therefore a two-term causal chain rather than a trilemma,
\begin{equation}
\boxed{
\begin{aligned}
&\text{physical fidelity of }\M\\
&\quad\xrightarrow{\text{Cholesky}}\text{density of }\A\\
&\quad\xrightarrow{\text{Pauli}}\text{measurement cost}.
\end{aligned}}
\label{eq:trilemma}
\end{equation}
in which only the physical fidelity of $\M$ and the resulting measurement cost are genuinely in opposition; the intermediate density of $\A$ is induced, and is further modulated by DOF ordering and coefficient thresholding. Lumped mass may be cheaper quantum mechanically but can perturb high modes; consistent mass may be more faithful but measurement-heavy. A credible study must report both modal discretization error and quantum resource consequences. For the chains studied here the discretization penalty is small precisely where the quantum method operates. Taking consistent mass as the reference, the lumped model reproduces the first four frequencies to within 3.88\% at $N=16$, 0.25\% at $N=64$, and 0.06\% at $N=128$; the disagreement grows with mode index, reaching 4.9\% at mode $N/4$ and 42\% at the highest mode. Since the finite-shot RMSE of the fundamental is 2.45\% at $10^5$ shots per Pauli term (\cref{tab:shotfreq}), the extra fidelity bought by consistent mass at $N=128$ is roughly forty times smaller than the measurement noise it must be read through, while costing $2950$ Pauli terms against $576$. Under the whitened, independently measured $10^5$-shots-per-term configuration tested here, the extra low-mode fidelity of consistent mass is therefore not justified by its measurement cost---an argument for the lumped operator and its $O(N\log N)$ term count of \cref{prop:lumped}. That conclusion is specific to the whitening and measurement strategy tested: \cref{sec:genquotient} shows that avoiding whitening cuts the representation cost of consistent mass to $750$ terms at $N=128$, against $2950$ whitened.

As an illustration, consider the $N=128$, $k=4$ consistent-mass chain of \cref{tab:scaling}. \Cref{tab:concrete} contrasts the two routes with order-of-magnitude figures.

\begin{table}[t]
\centering
\caption{Illustrative end-to-end comparison for the $N=128$, $k=4$ consistent-mass chain. Preprocessing is finite-element assembly (classical) versus Cholesky factorization, whitening, and Pauli decomposition (quantum). Quantum figures are order-of-magnitude estimates, not measured device results.}
\label{tab:concrete}
\begin{tabular}{lll}
\toprule
Quantity & Classical & Quantum\\
\midrule
Qubits & --- & 7\\
Preprocessing & $O(N)$ & $O(N^2\log N)$\\
Nonzero Pauli terms & --- & 2950\\
Core work & ${\sim}1400$ matvecs & ${\sim}10^{11}$ circuits$^{\ddagger}$\\
Wall time & ${<}1$\,s & ${\sim}3.5\times10^{5}$\,s\\
\bottomrule
\end{tabular}

\vspace{0.35em}
\parbox{\columnwidth}{\footnotesize $^{\ddagger}$${\sim}10^{5}$ shots per term $\times\,2950$ terms $\times$ a few hundred optimizer evaluations $\times\,4$ modes; the wall time assumes an optimistic $1\,\mu$s per circuit execution with no queueing, calibration, or readout overhead.}
\end{table}

The gap spans several orders of magnitude and is dominated by measurement, not preprocessing, which supports the qualitative statement that classical methods are overwhelmingly preferable for the present benchmarks. Any regime in which a quantum route could become competitive lies outside the variational setting studied here: it would require either fault-tolerant quantum phase estimation rather than sampling-based VQE, or operators that are genuinely dense so that classical sparsity offers no advantage---neither of which applies to sparse structural finite-element models. This is stated only to delimit scope and is not a claim of attainable advantage.

\subsection{Implications for structural health monitoring}

The finite-shot damage result changes the interpretation of quantum modal strain energy. The expectation-value identity is mathematically valid and practically reachable, but only at a measurement cost that must be stated explicitly. Detection requires resolving a small difference between two noisy ratios of expectation values. For baseline and damaged estimates $\widehat a_d/\widehat b_d$ and $\widehat a_0/\widehat b_0$, a first-order variance is
\begin{equation}
\begin{aligned}
\Var(\widehat\eta_d-\widehat\eta_0)
&\approx \mathbf g_d^T\boldsymbol\Sigma_d\mathbf g_d
+\mathbf g_0^T\boldsymbol\Sigma_0\mathbf g_0,\\
\mathbf g&=\begin{bmatrix}1/b&-a/b^2\end{bmatrix}^T.
\end{aligned}
\label{eq:ratiovar}
\end{equation}
For equal shots $n$ per nonidentity Pauli term, write $\Var(\widehat D_e)\approx V_{D,e}/n$. A normal-approximation design criterion for detecting a nonzero score at two-sided type-I error $\alpha$ and power $1-\beta$ is
\begin{equation}
\boxed{
n_{\mathrm{det}}\gtrsim
\frac{(z_{1-\alpha/2}+z_{1-\beta})^2V_{D,e^\star}}
{D_{e^\star}^2}.
}
\label{eq:shotsdetect}
\end{equation}
For localization, let $e^{(2)}$ be the strongest intact competitor and $\Delta_D=D_{e^\star}-D_{e^{(2)}}$. A conservative Bonferroni pairwise requirement is
\begin{equation}
\boxed{
n_{\mathrm{loc}}\gtrsim
\frac{(z_{1-\alpha/(n_e-1)}+z_{1-\beta})^2V_{\Delta_D}}
{\Delta_D^2},
}
\label{eq:shotslocalize}
\end{equation}
where $V_{\Delta_D}$ is the per-shot variance coefficient of the score margin. The absolute-value signs are linearized at their exact signs. These expressions are design approximations, not finite-sample guarantees. They are also matched to the alternative rather than to the null: $D_e$ is a sum of absolute values, so under $D_e=0$ the estimator is folded and biased upward by $\sigma\sqrt{2/\pi}$ per term, and the nominal type-I error is not attained. A deployed test should calibrate its null by permutation rather than from \cref{eq:shotsdetect}.

The calculation uses the first four modes with equal weights, independent baseline/damaged numerator and denominator estimates, and independent element-score variances. It is conservative because beneficial shared-denominator covariance is ignored. For 10\% damage, $10^5$ shots per term give a localization signal-to-noise ratio of 1.44, consistent with the observed 84.5\% success. Approximately $8.7\times10^4$ shots per term are required to detect the damaged-element score with 80\% power, whereas family-wise localization against all five intact competitors requires roughly $4.8\times10^5$ per term. With 648 nonidentity operator-term settings across four modes, six elements, and two structural states, the latter corresponds to approximately $3.1\times10^{8}$ circuit shots under the independent baseline---four modal states measured to part-in-$10^3$ precision for a single binary maintenance decision. That figure counts every (operator, term) pair as a distinct setting and is therefore an upper bound. Each element operator's Pauli support is contained in the 19 nonidentity strings of $\widetilde\A$ itself, exactly so for the chains at $N\le16$ and to within 12\% of the union at $N=32$, and $\sum_e\Hq_e=\widetilde\A$; one measurement of those 19 strings per mode and structural state therefore serves all seven operators. Counting distinct settings instead gives $8\times19=152$ measurement settings rather than 648, a factor of $4.3$ reduction in the number of distinct circuits that must be prepared. The corresponding shot total is not obtained by rescaling $4.8\times10^{5}$ shots per term, because re-using the records introduces exactly the numerator--denominator covariance that \cref{eq:ratiovar} discards; a covariance-aware power calculation is required and is left to future work. What the support structure establishes is the reduction in distinct settings, not a proportional reduction in total shots. Measurement therefore remains the dominant end-to-end cost even where the decision itself is reachable, which makes covariance-aware joint measurement and direct estimation of change observables the most promising levers on that cost.

\subsection{End-to-end quantum advantage remains unestablished}

Any future advantage claim must include:
\begin{enumerate}[label=(\alph*)]
\item assembly and reduction of $\M$ and $\K$;
\item factorization or generalized-quotient preparation;
\item Hamiltonian construction and Pauli decomposition;
\item state preparation and repeated measurements;
\item optimizer evaluations and higher-mode deflation;
\item readout of decision-relevant structural observables;
\item comparison with sparse classical eigensolvers at matched accuracy.
\end{enumerate}
For the present benchmarks, classical methods are overwhelmingly preferable.

\section{Limitations and research agenda}
\label{sec:limitations}

The main structural models remain one-dimensional chains, the most favorable case for Cholesky fill; \cref{sec:dimensions} reports an initial two- and three-dimensional scalar-mesh generalization confirming that densification does not improve, and appears later in $N$, as dimension increases under the locality-preserving ordering studied here. That study nonetheless uses regular scalar grids, not the block DOFs, irregular topology, constraints, and material heterogeneity of real bridges and buildings, and the measured exponents must not be extrapolated to such structures or to alternative DOF orderings and block-sparse formulations. Pauli decomposition is evaluated only under binary amplitude encoding. The finite-shot study uses independent Pauli measurements and exact input states; it does not include circuit noise, state-preparation error, grouped covariance, or device calibration. The damage study uses stiffness-only damage and no environmental variability; long-term monitoring campaigns show that environmental effects can rival or exceed damage-induced modal changes \cite{PeetersZ24_2001}; damage identification on that same campaign relied on modal curvature rather than on a strain-energy index of the kind used here \cite{Maeck2001}. It also supplies the damaged mode shapes as exact eigenvectors of a known damaged operator, so it isolates measurement variance and does not model the identification problem: in deployment the damaged states would be obtained from measured response data, not from an eigensolve of the very model being inferred, whereas the element operators $\Hq_e$ of \cref{eq:He} are correctly built from the baseline model alone. Cholesky timings are platform-specific, while the structural sparsity counts are more portable.

The next high-impact study should combine:
\begin{enumerate}[label=(\roman*)]
\item a laboratory frame or bridge dataset with FDD/SSI uncertainty;
\item block-sparse three-dimensional finite elements and alternative DOF orderings;
\item Cholesky, generalized quotient, and domain-decomposed mappings;
\item grouped or shadow-based measurement of modal observables;
\item near-degenerate subspace tracking under temperature and damage;
\item a matched end-to-end comparison against ARPACK/Lanczos and reduced-order classical models.
\end{enumerate}

\section{Conclusions}
\label{sec:conclusion}

This paper establishes a stricter evaluation standard for variational quantum modal analysis. The mass-orthonormal Cholesky map is physically consistent, but its preprocessing and densification costs are not negligible. Gershgorin bounds can certify padding without fixing the optimizer scale. Higher-mode penalties admit an exact threshold, and approximate-projector perturbation theory shows that both undersized and oversized penalties degrade the available robustness certificate, and locates the penalty that minimizes it. Variance and leakage should be diagnostics or constrained refinements, not arbitrarily weighted additions to the primary energy. Temporal mode identity should be assigned after independent eigensolution so that continuity cannot mask damage.

The expanded experiments show both promise and difficulty. Increasing circuit depth recovers small exact-state modal problems, while a shallow ansatz fails systematically despite sufficient qubit dimension. In the tested chains, consistent mass substantially increases Hamiltonian density and Pauli count, and an initial two- and three-dimensional generalization shows this effect worsening, not improving, with spatial dimension under the locality-preserving ordering studied here---an ordering the paper also shows to matter more than the dimension does, since a fill-reducing permutation raises the Pauli count by up to an order of magnitude and a scrambled one saturates the worst case exactly. That generalization also validates a subspace diagnostic on a genuine degenerate pair unavailable in one dimension. Finite-shot frequency and damage estimates remain costly---the analytic design estimate for family-wise localization of a 10\% stiffness loss at 80\% power is ${\sim}10^{8}$ circuit executions---with the lowest mode particularly sensitive to sampling variance. The 1D Pauli growth is moreover decelerating rather than settling onto a power law---the local exponent falls to $1.25$ over the last tested doubling---so the fitted exponents throughout should be read as illustrative fits, not asymptotic laws. These results do not establish quantum advantage. They identify the mathematical conditions and resource bottlenecks that a future structural quantum algorithm must overcome.

Three conclusions are durable. First, logarithmic qubit storage does not imply logarithmic circuit or measurement cost: the encoding is compact, the state preparation and the observable are not. Second, deflation can be certified spectrally, but approximate projectors and small gaps make penalty selection and the block treatment of near-degenerate pairs consequential rather than incidental. Third, on the variational structural benchmarks studied here, measurement and decision precision remain the binding resource and no quantum advantage is established---a statement about this algorithmic regime, not about the fault-tolerant, qubitization and Hamiltonian-simulation routes surveyed in \cref{sec:intro}. The main conclusion is therefore precise: qubits are necessary to encode the state, but spectrally valid modes, expressive circuits, sparse observables, and decision-level measurement precision are necessary to solve the engineering problem.

\section*{Data and code availability}
The accompanying archive contains the complete LaTeX source, bibliography, generated figures, CSV results, and Python scripts required to reproduce every reported number. This covers the exact-state, scaling, optimizer-restart, certified-penalty, trainability, depth-sweep, finite-shot and damage-localization studies, the two- and three-dimensional generalization and degeneracy diagnostic of \cref{sec:dimensions}, and the oracle-convention sensitivity case of \cref{rem:convention}. A validation script re-executes a deterministic regression suite into temporary directories and diffs the regenerated tabular outputs against the committed data at zero tolerance, rather than checking committed files against hard-coded expectations. Ten datasets are regenerated and diffed by default and five more under \texttt{-{-}full}, which additionally reruns the Monte Carlo studies; the remaining manuscript-quoted values are bound by assertion to their generated summaries rather than by independent regeneration. The script was itself fault-injection tested, catching 7 of 7 deliberately introduced defects with no false positive on the clean package. Bit-level reproducibility is claimed only for the environment recorded in \texttt{environment\_versions.txt} and \texttt{requirements-lock.txt} (Python 3.12.3, NumPy 2.3.5, SciPy 1.17.0, pandas 2.2.3, Matplotlib 3.10.8, and pdf\TeX{} from \TeX{} Live 2026). The deterministic studies depend on the BLAS in use and on pandas' stable-sort handling of exactly tied coefficients, so the results are expected to reproduce under that environment and are not claimed to be bit-identical across platforms.

\section*{Conflict of Interest}
The authors declare no conflict of interest.

\appendices
\section{Element matrices and the scope of the models used}
Assembly follows \cref{eq:assembly} from standard element stiffness and mass matrices; for a two-dimensional truss element these are the usual $4\times4$ direction-cosine forms, and for an Euler--Bernoulli beam the usual $4\times4$ bending and consistent-mass forms. None is used here. Every model solved in this paper is a scalar spring--mass chain or grid with one degree of freedom per node, chosen so that the comparison isolates spatial dimension and mass model from the confound of degrees of freedom per node. Extending the resource accounting to vector-valued elements changes $N$ per node and therefore $n_q$, but not the mapping, the deflation theory, or the measurement analysis.

\section{Sequential modal algorithm with provable deflation thresholds}
\begin{algorithm}[t]
\caption{Mass-orthonormal sequential VQE with provable deflation thresholds}
\label{alg:certvqe}
\begin{algorithmic}[1]
\Require Reduced $\M\succ0$, $\K=\K^T$; target modes $m$; ansatz family; tolerances.
\State Compute $\M=\mathbf L\mathbf L^T$ using a sparse ordering where appropriate.
\State Form $\A=\mathbf L^{-1}\K\mathbf L^{-T}$ by triangular solves.
\State Compute certified $U_G$ and choose optimization scale $s_t$.
\State Form $\Hq_p$ with $\alpha=U_G/s_t+\delta_p$.
\State Decompose the required operators and group compatible measurements.
\For{$r=1,\ldots,m$}
  \State Select and record the bound policy: $U_r=U_G/s_t$, or compute $U_r=\vartheta_r^{(p)}+\delta_{\mathrm{fp}}$ by a $p$-dimensional Ritz/Lanczos pass and record $C_{\mathrm{bound}}$.
  \For{$j<r$}
    \State Obtain $\delta_j$ from exact validation data or the declared residual-plus-shot confidence rule in \cref{eq:deltabound}.
    \State Set $\beta_j=U_r-\widehat\varepsilon_j+\delta_j+\tau$.
  \EndFor
  \State Minimize $C_r(\theta)=\langle\psi(\theta)|\widehat\Hq_r|\psi(\theta)\rangle$ over multiple initializations.
  \State Optionally solve the constrained variance refinement (supplementary material).
  \State Reject or flag the state if any component of \cref{eq:diagnosticvector} exceeds its declared tolerance.
  \State Recover $f_r$ and $\ph_r$ using \cref{eq:freqrecover,eq:recover}.
\EndFor
\end{algorithmic}
\end{algorithm}

For the reproducible six-DOF exact-state validation, the declared acceptance set was
\begin{equation}
\begin{aligned}
\epsilon_f&\le10^{-4}\%, &
\rho&\le10^{-5}, &
\sigma_H^2&\le10^{-10},\\
\ell_p&\le10^{-8}, &
o&\le10^{-8}, &
1-\MAC^{(M)}&\le10^{-6}.
\end{aligned}
\label{eq:exampletolerances}
\end{equation}
These values are numerical verification tolerances, not universal SHM thresholds; several are unattainable at any sampling budget, the variance bound $\sigma_H^2\le10^{-10}$ most obviously so, and $\langle\Hq_p^2\rangle$ carries its own measurement cost that \cref{eq:resourcevector} does not itemize. They accept all four depth-2 exact-state modes and reject the depth-1 results. Field applications must replace them with tolerances calibrated to model discrepancy, environmental variability, and decision consequences.

Algorithm~\ref{alg:certvqe} was executed on this benchmark under the quantum-only Gershgorin policy \cref{eq:globalUr}, so that no auxiliary classical eigensolve enters the procedure. With $s_t=U_G$ the certified bound is $U_r=1$, the radius $\delta_j$ is the residual $\sigma_{H,j}$ returned by the run itself, and $\tau=10^{-3}$. At depth 2 the resulting penalties lie between $0.68$ and $0.98$: they satisfy the threshold of Theorem~1 by construction rather than by tuning, and they adapt to solution quality as intended---at depth 1, where the extracted modes are poor and $\sigma_{H,j}$ is correspondingly large, the same rule returns $0.81$ to $1.02$, automatically widening the margin to absorb the worse projectors. They are also three to four and a half times smaller than the ad hoc common penalty $\beta=3$ used for comparison. Measured against \cref{eq:betastar} rather than against the bare threshold, the Gershgorin policy places the excess $U_r-\varepsilon_r+\delta_j+\tau$ at $2.0$, $2.9$ and $4.5$ times $G_r$ for targets $r=4,3,2$, i.e. always past saturation and never in the divergent regime, at the cost of a Davis--Kahan certificate loosened by factors of $1.3$, $1.8$ and $3.2$ relative to the optimum. Conservatism in the bound policy is what buys this margin. At depth 2 the certified run recovers all four modes with frequency error below $10^{-12}$\%, mass-weighted MAC equal to one to twelve digits, prior-mode overlap below $4\times10^{-15}$, and padding leakage below $10^{-15}$, meeting the acceptance set above with several orders of margin; at depth 1 it fails for the ansatz reasons of \cref{rem:depth}, not because of the penalty rule.

\section*{Acknowledgment}
In accordance with IEEE policy on the disclosure of AI-generated content, the authors record the following. An AI assistant (Claude, Anthropic) was used during the revision of this article: initial drafts of \cref{sec:genquotient,sec:ordering,sec:neardegenerate} and of revised passages elsewhere, together with the corresponding analysis and validation scripts, were AI-generated and subsequently edited; publication records for several references were retrieved and checked with AI assistance; and an AI-assisted audit of the article's derivations and numerical claims was performed, with its findings assessed before adoption. All definitions, theorems, proofs, numerical results, and conclusions were checked by the authors, who take full responsibility for the content of this article. No AI system is an author of this work.

\bibliographystyle{IEEEtran}
\bibliography{references}

\ifdefined\ACCESS\EOD\fi
\end{document}


\title{Supplementary Material: Qubits Are Not Enough --- Deflation and Measurement
Limits in Variational Quantum Modal Analysis}
\maketitle

\noindent This supplement contains three procedures referenced from the main article.
Each is developed and analyzed but \emph{not} evaluated experimentally, which is why it
is placed here rather than among the article's results. Equation and section numbers are
prefixed with \textsf{S}; all references to the main article are by name.

\section{Lexicographic variance refinement}
\label{app:lexico}
Adding $\gamma\sigma_H^2+\zeta\ell_p$ to the energy requires a conversion between quantities with no common
scale. The alternative used here states the trade-off as constraints at declared tolerances:
\begin{equation}
\begin{aligned}
&\min_{\theta}\ \sigma_H^2(\theta)\\
&\text{s.t. }C_r(\theta)\le C_r^\star+\epsilon_E,\qquad
\ell_p(\theta)\le\epsilon_p,\\
&\hspace{2.0em}\max_{j<r}|\langle\widehat u_j|\psi(\theta)\rangle|^2\le\epsilon_O.
\end{aligned}
\label{eq:lexico}
\end{equation}
The energy ordering is preserved within $\epsilon_E$ rather than traded against variance at an
application-defined rate. This refinement is not exercised in any experiment reported here.

\section{Exact relation between projector error and state angle}
For normalized vectors $u$ and $\widehat u$, the rank-one projector difference satisfies
\begin{equation}
\norm{\widehat u\widehat u^T-uu^T}_2=\sin\angle(\widehat u,u).
\label{eq:projectorangle}
\end{equation}
Therefore the deflation perturbation may be written as
\begin{equation}
\norm{\Delta_r}_2\le
\sum_{j<r}\beta_j\sin\angle(\widehat u_j,u_j),
\label{eq:anglepert}
\end{equation}
which gives a direct rule for balancing penalty margin against the quality of previously extracted modes.

\section{Post-hoc temporal mode tracking}
\label{app:tracking}
Every monitoring epoch is solved independently using the same independently defined objective. A previous solution may initialize the circuit, but it does not enter the objective. After solution, define for candidate pair $(i,j)$
\begin{equation}
\mathbf z_{ij,t}=
\begin{bmatrix}
\log(f_{j,t}/f_{i,t-1})\\
1-\MAC_{ij,t}^{(M)}
\end{bmatrix},
\qquad
c_{ij,t}=\mathbf z_{ij,t}^T\boldsymbol\Sigma_{i,0}^{-1}\mathbf z_{ij,t},
\label{eq:trackingcost}
\end{equation}
where $\boldsymbol\Sigma_{i,0}$ is estimated from healthy baseline variability. MAC-based comparison of identified modes is standard in operational modal analysis \cite{Reynders2012}; the assignment formulation used here is stated as this paper's own construction rather than as established practice, and what is adapted is the uncertainty normalization and the insistence that pairing occur strictly after an independent eigensolve. Mode labels are assigned by
\begin{equation}
\pi_t^*=\argmin_{\pi}\sum_i c_{i,\pi(i),t}.
\label{eq:hungarian}
\end{equation}
A change alarm is based on $c_{i,\pi(i),t}$ exceeding a baseline quantile. A $\chi^2$ calibration is not appropriate here: $1-\MAC$ is non-negative and concentrates at zero, so the quadratic form is not $\chi^2_2$ distributed and the empirical quantile should be used instead. Because pairing occurs after independent eigensolution, the objective cannot pull a damaged mode toward its healthy predecessor. This is weaker than unbiasedness of the alarm: the assignment itself can still mislabel under damage-induced mode veering, after which the change statistic is referenced to the wrong baseline mode, and modes that go unidentified in an epoch---routine in operational modal analysis---need an explicit missing-label rule that \cref{eq:hungarian} does not provide. A further practical restriction is that $\MAC^{(M)}$ requires $\M$ at every degree of freedom, whereas output-only identification returns shapes only at instrumented coordinates; a sensor-space unweighted MAC must be substituted there, at the cost of the mass-orthogonality interpretation. Near degeneracy, the same procedure is applied to principal angles between subspaces.

This rule is not exercised in any experiment reported here.

\bibliographystyle{IEEEtran}
\bibliography{references}